\documentclass{article}
\usepackage[a4paper, total={6.5in, 10in}]{geometry}
\usepackage[utf8]{inputenc}
\usepackage[dvipsnames]{xcolor}
\usepackage{amsfonts}
\usepackage{amsthm}
\usepackage{amsmath}
\usepackage{parskip}
\usepackage{mathrsfs}
\usepackage{amssymb}
\usepackage{xtab}
\usepackage{setspace}
\usepackage{datetime}
\usepackage{svg}
\usepackage{csquotes}
\usepackage{fancyhdr}
\usepackage{float}
\usepackage{comment}
\usepackage{enumitem}

\usepackage{amssymb,mathalfa,float,array,titlesec,mdframed,listings,amsmath,booktabs,dsfont,textcomp,graphicx,wrapfig,cases}
\usepackage[font=small,labelfont=bf]{caption}

\usepackage[utf8]{inputenc}
\usepackage[english]{babel}

\usepackage{float}

\usepackage{placeins}
\usepackage{subcaption}
\usepackage{hyperref}
\usepackage{cleveref}

\usepackage{cite}

\usepackage{multirow}

\graphicspath{{Figures/}}

\allowdisplaybreaks

\usepackage{thmtools}
\usepackage{thm-restate}

\declaretheorem[name=Theorem,numberwithin=section]{theorem}
\declaretheorem[name=Definition,numberwithin=section]{definition}
\declaretheorem[name=Remark,numberwithin=section]{remark}

\declaretheorem[name=Assumption]{assumption}
\declaretheorem[name=Example,numberwithin=section]{example}

\newcounter{subassumption}

\usepackage{authblk}
\title{On evolving network models and their influence on opinion formation}
\author[1,2]{Andrew Nugent}
\author[2]{Susana N. Gomes}
\author[2]{Marie-Therese Wolfram}
\affil[1]{MathSys CDT, University of Warwick}
\affil[2]{Mathematics Institute, University of Warwick}

\begin{document}

\maketitle

\begin{abstract}
     In this paper, we propose a new model for continuous time opinion dynamics on an evolving network. As opposed to existing models, in which the network typically evolves by discretely adding or removing edges, we instead propose a model for opinion formation which is coupled to a network evolving through a system of ordinary differential equations for the edge weights. We interpret each edge weight as the strength of the relationship between a pair of individuals, with edges increasing in weight if pairs continually listen to each other’s opinions and decreasing if not. We investigate the impact of various edge dynamics at different timescales on the opinion dynamics itself. This is done partly through analytic results and partly through extensive numerical simulations of two case studies: one using bounded confidence interaction dynamics in the opinion formation process (as in the classical Hegselmann-Krause model) and one using an exponentially decaying interaction function. We find that the dynamic edge weights can have a significant impact on the opinion formation process, since they may result in consensus formation but can also reinforce polarisation. Overall, the proposed modelling approach allows us to quantify and investigate how the network and opinion dynamics influence each other.
\end{abstract}


\textbf{Highlights:}
\begin{itemize}
    \item We propose a novel system of coupled ODEs for opinion and network evolution. 
    \item Numerical case studies explore the impact of network dynamics on opinion formation. 
    \item Some network dynamics promote consensus while others lead to rapid local clustering.
    \item Limiting models exist at extreme relative timescales of opinion and network dynamics.
\end{itemize}

\section{Introduction} \label{Section: Introduction}

The study of opinion dynamics began with a model of how a small group could reach consensus on a single issue \cite{degroot1974reaching} and has since expanded to complex models describing opinion formation across populations and societies \cite{brooks2020model,hegselmann2015opinion,ayi2021mean}. Opinion formation models have been used to forecast opinions \cite{de2016learning}, study financial markets \cite{zha2020opinion}, design opinion controls \cite{fornasier2014mean} and methods to prevent the disruption of consensus \cite{gaitonde2020adversarial}. 

A key idea proposed by Hegselmann and Krause \cite{hegselmann2002opinion}, also introduced separately by Deffuant \textit{et al.} \cite{deffuant2000mixing}, is the concept of bounded confidence: that individuals only interact when their opinions are already sufficiently close. More generally, many models employ an interaction function that determines the extent to which individuals value each others' opinions based on the distance between those opinions. Such non-linear models are capable of displaying a variety of behaviours and have been used to study the influence of stubborn or radical individuals \cite{hegselmann2015opinion}, online media \cite{brooks2020model} as well as the transition between consensus and polarisation \cite{wang2017noisy,motsch2014heterophilious}. 

An increasingly common feature in opinion dynamics models is the introduction of an underlying (social) network across which individuals interact \cite{amblard2004role,gabbay2007effects,lorenz2007continuous}. This approach is more representative of real opinion formation, as individuals could not be expected to interact with the entirety of a large population, nor to consider all opinions equally. A further development is to allow this underlying network to evolve in response to the opinion formation process happening across it. A common approach is to evolve the network discretely by repeatedly `rewiring' the network: removing one edge from the network and adding a new edge somewhere else \cite{gross2008adaptive, kozma2008consensus, iniguez2009opinion, barre2021fast, kan2023adaptive}. Typically this is done according to the idea of `homophily': edges connecting individuals with strongly differing opinions are removed, whereas new edges are added between individuals whose opinions are more similar \cite{chu2022non,kan2023adaptive}. An alternative approach is to record some additional parameter that controls when edges are switched on or off \cite{mariano2020hybrid}. Over time, this creates community structure within the network, with individuals within communities having similar opinions and different communities having differing opinions \cite{iniguez2009opinion}. Another interesting feature is the inclusion of triadic closure: adding an edge between two nodes that both have a connection to a shared third node \cite{iniguez2009opinion}. This procedure is particularly relevant given the interpretation of the network as a social network, as the two nodes that become connected are the `friend-of-a-friend'(FOAF) \cite{rapoport1953spread,heider1982psychology}. However, with few exceptions \cite{ayi2021mean,evans2018opinion}, updates to the underlying network are typically discrete and edges are either present or absent. 

The goal of this paper is to introduce a new model in which the underlying network represents the strength of the relationship between two individuals and evolves continuously alongside their opinions. The idea that continued interaction both builds social ties and is required to maintain them motivates the network dynamics we consider. We focus on how the evolution of the network affects opinion formation, especially when it brings the population closer to, or further from, consensus. Our `dynamic network' model provides a new approach in modern continuous-time opinion dynamics by treating individuals' relationships as an active part of the opinion formation process.  

Due to the vast array of opinion models that have been proposed, with variation in discrete/continuous time, discrete/continuous opinion spaces, deterministic/stochastic processes, finite/infinite population sizes and differing communication regimes, we begin by motivating the type of model we will consider and describe its behaviour on a fixed network in Section \ref{Section: Fixed network}. This provides a baseline against which the dynamic network model defined in Section \ref{Section: Dynamic network} can be compared. Several analytic results about this model are given in Section \ref{Section: General properties}. Three example weight dynamics are then introduced in Section \ref{Section: examples} and studied through extensive simulation in Sections \ref{Section: BC} and \ref{Section: exp}. In Section \ref{Section: Timescales} we discuss the impact of altering the relative timescale of opinion formation and network evolution and consider the limits as these timescales become extreme. Finally Section \ref{Section: Discussion} interprets our results and discusses potential future research. 

\section{Fixed network} \label{Section: Fixed network}

For a finite, fixed population of $N$ individuals, we consider an opinion formation model given by a system of $N$ ordinary differential equations (ODEs). Such a formulation avoids introducing an arbitrary fixed timestep or a fixed number of possible opinions. Additionally, using a system of ODEs will allow us to analyse the effect of extreme differences in timescales of the opinion formation process and network dynamics in Section \ref{Section: Timescales}. A finite population is chosen to give clear meaning to the idea of a social network. The opinion formation process is one of continuous weighted averaging of opinions, inspired by the models of DeGroot \cite{degroot1974reaching}, Deffuant and Weisbuch \cite{deffuant2000mixing} and Hegselmann and Krause (HK) \cite{hegselmann2002opinion}. Although bounded confidence models will serve as an important example throughout this paper, the model is presented for a general interaction function. 

Let $x(t) = \big( x_1(t), \dots, x_N(t) \big)$ be the individuals' opinions at time $t\geq0$. Assume $x_i(0) \in [-1,1]$ for all individuals $i$ (although in principle any bounded interval could be used). Here $x_i$ can be interpreted as the extent to which individual $i$ agrees or disagrees with a given statement. Define an interaction function $\phi:[0,2] \rightarrow [0,1]$ so that $\phi\big(|x_i - x_j|\big)$ describes the value which individual $i$ gives the opinion of individual $j$. We take $|\cdot|$ to be the standard Euclidean distance, but in principle, any metric could be used. We make the following assumptions about the interaction function $\phi$. 

\begin{assumption} \label{Assumption group: phi}
The interaction function $\phi\colon [0,2] \rightarrow [0,1]$ satisfies 
\begin{enumerate}[label=\alph*)]
    \item
       $\phi(r) \geq 0$ for all $r\in[0,2]$. \label{Assumption: phi positive}
    \item 
        $\phi$ is integrable over $[0,2]$. \label{Assumption: phi integrable}
\end{enumerate}
\end{assumption}

In general, we do not assume that $\phi$ is continuous. For example, the classical HK model introduces a confidence level $R\geq0$ and the (discontinuous) interaction function
\begin{equation} \label{Eqn: confidence indicator}
    \phi_R(|x_i - x_j|) = 
    \begin{cases}
        1 & \text{if  } |x_i - x_j| < R \\
        0 & \text{if  } |x_i - x_j| \geq R.
    \end{cases}
\end{equation}
Next we state some definitions and assumptions used to define and analyse the proposed opinion formation models on networks. Throughout this paper all networks will use the same node set: $V = \{1,\dots,N\}$. That is, one node corresponding to each individual in the population. With this node set fixed, the network can be characterised by its weighted adjacency matrix $w\in[0,1]^{N\times N}$, hence we refer to $w$ as a network throughout. We interpret the network as describing the relationships between individuals, their level of trust, and familiarity.

\begin{definition}
    A network $w$ is called undirected if $w_{ij} = w_{ji}$ for all individuals $i,j \in \lbrace 1, \ldots N\rbrace$, $i \neq j$ and is called directed otherwise.
\end{definition}
\begin{definition}
    Fix a network $w$. Nodes $i \in \lbrace 1, \ldots N\rbrace$ and $j \in \lbrace 1, \ldots N\rbrace$, $i \neq j$, are in the same connected component of the network if there exists some $m \in \mathds{N}$ such that $(w^m)_{ij} >0$. Furthermore, if $(w^N)_{ij} >0$ for all $i,j$ then the network is called connected. 
\end{definition}

\begin{definition}
    The (in-)degree of an individual will be denoted by $k_i$ and the mean degree by $\Bar{k}$. These are defined by:
    \begin{equation}
        k_i = \sum_{j=1}^N w_{ij}\, \text{ and }
        \Bar{k} = \frac{1}{N} \sum_{i=1}^N k_i\,.
    \end{equation}
\end{definition}
Individuals' opinions then evolve according to the system of ODEs,
\begin{equation} \label{Eqn: fixed network ODEs}
    \frac{dx_i}{dt} = \frac{1}{k_i} \sum_{j\neq i} w_{ij}\,\phi\big(|x_i - x_j|\big)\,(x_j - x_i), \quad\quad i = 1,\dots,N \,,
\end{equation}
so that individuals move towards the weighted mean opinion of those they interact with. The normalisation by $k_i$ ensures that individuals with a higher degree in $w$, that is more `popular' individuals, do not change their opinions faster than those individuals with a low degree. We refer to this system as the fixed network model. 

Models of the form \eqref{Eqn: fixed network ODEs}, with a general interaction function, have been studied previously in the context of opinion dynamics on fully connected networks \cite{motsch2014heterophilious}. Models including a fixed network have been considered in the case of bounded confidence interactions \cite{parasnis2018hegselmann,fortunato2005consensus}. The impact of social networks has also been studied in the Deffuant-Weisbuch model \cite{lorenz2007power,stauffer2004simulation}, in Bayesian models \cite{acemoglu2011opinion} and in a model combining HK with the DeGroot model \cite{wu2022mixed}. 

\begin{remark}
    Although we focus here on opinions on a single topic, the separation of the social network $w$ from the interaction function $\phi$ allows for a straightforward extension into multiple topics, with the social network $w$ common among all topics. Indeed the definition here, with a general metric used to calculate the interaction function $\phi\big(|x_i - x_j|\big)$, already applies if each individual's opinion $x_i$ has multiple dimensions. 
\end{remark}

We now show that the fixed network model is well-defined in the context of opinion dynamics. 

\begin{definition}
The diameter $D: [0,T] \rightarrow [0,2]$ of the population's opinions is defined as
    \begin{equation}\label{Eqn:diameter}
        D(t) := \max_{ij} |x_i(t) - x_j(t)|.
    \end{equation}
    Furthermore, we say that the population reaches consensus if $D(t)\rightarrow0$ as $t\rightarrow\infty$. 
\end{definition}
\begin{restatable}{proposition}{boundedness (fixed network)}  \label{Prop: Boundedness (fixed network)}
    Let $\phi$ satisfy Assumption \ref{Assumption group: phi} and fix initial opinions $x(0)\in[-1,1]^N$. Then the solution $x(t)$ to \eqref{Eqn: fixed network ODEs} satisfies $x(t)\in[-1,1]^N$ for all times $t\geq0$. Furthermore, $\lim_{t\rightarrow\infty} D(t)$ exists and lies in the interval $[0,2]$. 
\end{restatable}
\begin{proof}
    As the population size $N$ is fixed and finite we can identify at any time the maximum and minimum opinions held in the population. Let $i$ be an individual holding the minimum opinion at time $t$, then $x_j \geq x_i$ for all $j\in\{1,\dots,N\}$ so
    \begin{equation}
        \frac{dx_i}{dt} = \frac{1}{k_i} \sum_{j\neq i} w_{ij}\,\phi\big(|x_i - x_j|\big)\,(x_j - x_i) \geq 0. 
    \end{equation}
    Similarly, any individual holding the maximum opinion at time $t$ will have a non-positive derivative. Thus the maximum opinion held in the population cannot increase and the minimum opinion held cannot decrease. So since $x_i(0)\in[-1,1]$ for all $i\in\{1,\dots,N\}$ we have that $x_i(t)\in[-1,1]$ for all $i\in\{1,\dots,N\}$, $t\geq 0$. More generally, opinions cannot leave any interval in which the initial opinions of all individuals lie. 

    In addition, this argument shows that $D(t)$ is a non-increasing function of time. Furthermore, it is bounded below by $0$, hence it must converge as $t\rightarrow\infty$. As $D(t)\in[0,2]$ for all $t\geq0$, the limit also lies in this interval. 
\end{proof} 

Next, we present a condition under which consensus in the fixed network model can be guaranteed. Lemma \ref{Lemma: MotschTadmor} below provides a useful intermediary result, the proof of which is a minor modification of the energy argument presented in Theorem 4.3 of \cite{motsch2014heterophilious} and is given in Appendix \ref{Appendix: Proofs}. 
\begin{restatable}{lemma}{MotschTadmor} \label{Lemma: MotschTadmor}

    Let Assumption \ref{Assumption group: phi} on $\phi$ be satisfied and assume $w\in[0,1]^{N \times N}$ is undirected. Fix $x(0)\in[-1,1]^N$ and let $x(t)$ be the solution to \eqref{Eqn: fixed network ODEs}. Then for any $\varepsilon > 0$ there exists a time $t^* < \infty$ such that
    \begin{equation*}
        \sum_{i,j} w_{ij}\,\phi\big(|x_i(t^*) - x_j(t^*)|\big)\,| x_j(t^*) - x_i(t^*)|^2 < \varepsilon\,.
    \end{equation*}
\end{restatable}
Broadly speaking, Lemma \ref{Lemma: MotschTadmor} states that for two individuals with $w_{ij}>0$ at some point in time either the distance between the two's opinion, or the interaction function at that distance, must be arbitrarily small. Therefore, if we eliminate the possibility of the interaction function becoming arbitrarily small, this would require all neighbours in $w$ to have `close' opinions. This motivates the following proposition. 
\begin{restatable}{proposition}{convergence} \label{Prop: Fixed network convergence}
     Let $w\in[0,1]^{N \times N}$ be a undirected and connected network and let $\phi$ satisfy Assumption \ref{Assumption group: phi}. Furthermore assume that there exists a constant $c > 0$ such that $\phi(r) > c$ for all $r\in[0,2]$. Then for any initial opinions $x(0)\in[-1,1]^N$, the opinion formation process in \eqref{Eqn: fixed network ODEs} reaches consensus. 
\end{restatable} 
\begin{proof} 
    Define $w_m$ to be the smallest non-zero weight in $w$, 
    \begin{equation}
        w_m = \min_{i,j}\, \{w_{ij} \colon w_{ij} > 0 \}. 
    \end{equation}
    Take any $\varepsilon > 0$. As $w$ is undirected, by Lemma \ref{Lemma: MotschTadmor} there exists a time $s$ such that 
    \begin{equation*}
        \sum_{i,j} w_{ij}\,\phi\big(|x_i(s) - x_j(s)|\big)\,|x_j(s) - x_i(s)|^2 < \frac{c \varepsilon^2 w_m}{N^2}\,.
    \end{equation*}
    As all the terms in this sum are positive and $\phi(r) > c$ for all $r\in[0,2]$, we have that for any pair of individuals $i,j$, the product
    \begin{equation*}
        w_{ij}\,|x_j(s) - x_i(s)|^2 < \frac{\varepsilon^2 w_m}{N^2}\,.
    \end{equation*}
    Hence for all $i,j$, either $w_{ij} = 0$ or 
    \begin{equation*}
        |x_j(s) - x_i(s)|^2 < \Big( \frac{\varepsilon}{N} \Big)^2\,.
    \end{equation*}
    As $w$ is connected there exists a path in $w$ of length at most $N$ between any two individuals. Hence at time $s$, since all neighbouring individuals' opinions are a distance strictly less than $\varepsilon/N$ apart, the diameter of the population's opinions $D(T) < \varepsilon$. As $D$ is non-increasing in time this means $D(t) < \varepsilon$ for all $t \geq T$, hence $D(t)\rightarrow0$ as $t\rightarrow\infty$ and the population reaches consensus. 
\end{proof}

Proposition~\ref{Prop: Fixed network convergence} not only provides a condition for consensus, but also gives some intuition about the causes of disagreement. If the population does not reach consensus then either the network $w$ must be directed or disconnected, or the interaction function $\phi$ must become arbitrarily small at some distance. 

\begin{figure}[ht!]
    \centering
    \includegraphics[width = \linewidth, trim = {1cm 2cm 1cm 2cm}, clip]{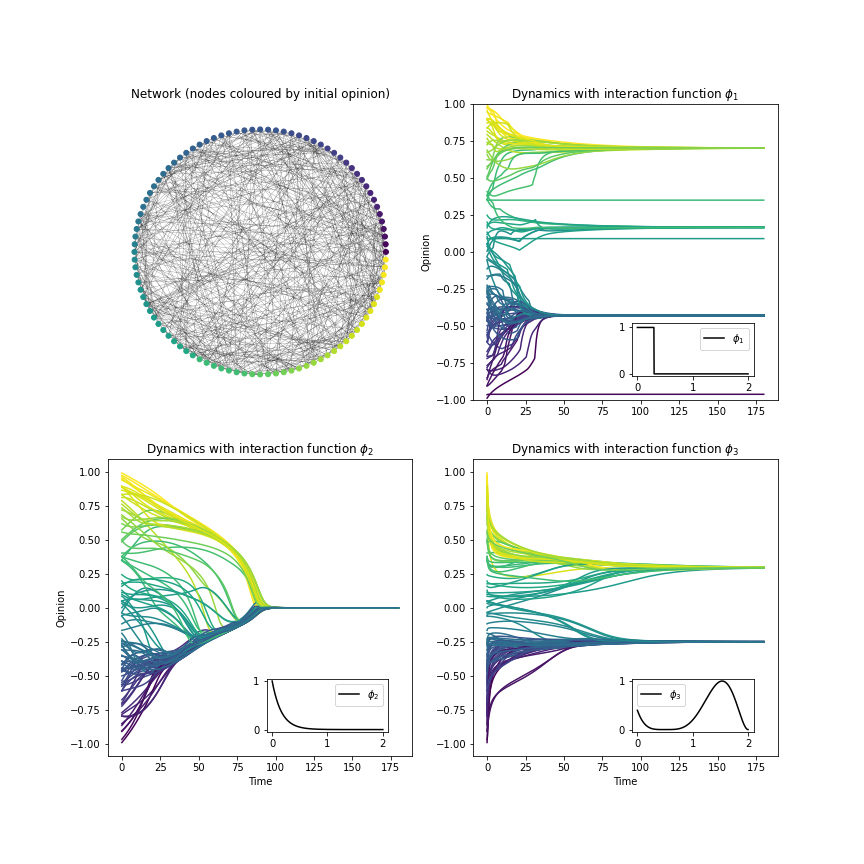}
    \caption{Example dynamics of the fixed network model. The top left panel shows the network, with nodes coloured by their initial opinions. The remaining panels show the opinion formation process over time for the three different interaction functions in \eqref{Eqn: example interaction functions}. The same network and initial conditions are used for each example. Each line represents the opinion of one individual over time and is coloured according to the individual's initial opinion.}
    \label{fig:fixed network example interaction functions}
\end{figure}

We now show some example dynamics for the fixed network model in Fig.\ref{fig:fixed network example interaction functions}. The top left panel shows a network of $N = 100$ nodes, with nodes coloured by their initial opinions in $[-1,1]$. Initial opinions were chosen from a uniform distribution on $[-1,1]$. The network is an Erd\H{o}s-R\'enyi random network with edge probability $p = 0.1$ \cite{erdHos1960evolution}. In this example, the network is connected. The remaining panels of Fig.\ref{fig:fixed network example interaction functions} show the opinion formation process over time for three different interaction functions, with the same network and initial conditions. The interaction functions used are:
\begin{equation} \label{Eqn: example interaction functions}
    \phi_1(r) = 
    \begin{cases}
        1 & \text{if  } r < \frac{3}{10} \\
        0 & \text{if  } r \geq \frac{3}{10} 
    \end{cases} ,\quad
    \phi_2(r) = e^{-6r} ,\quad
    \phi_3(r) = \frac{8}{5} \bigg(r - \frac{1}{2} \bigg)^4 \big(r+1 \big) \big(r-2 \big)^2 ,
\end{equation}
for $r\in[0,2]$ and are shown in insets in each panel. All three interaction functions satisfy Assumption \ref{Assumption group: phi}. Both $\phi_1$ and $\phi_2$ are decreasing functions, with the former being discontinuous and the latter continuous. In addition, $\phi_2$ is strictly positive over the entire domain, meaning Proposition \ref{Prop: Fixed network convergence} guarantees the consensus we observe. Conversely $\phi_1$ is zero for $r\in[\frac{3}{10},2]$, so many pairs of individuals do not interact. In fact, we observe some individuals who never change their opinions. As is typical in bounded confidence models \cite{hegselmann2015opinion}, a number of opinion clusters form. This clustering is also observed in the dynamics under interaction function $\phi_3$. Function $\phi_3$ allows for interaction between both close and distant individuals. Allowing interactions between distant individuals initially brings those with extreme opinions closer to the centre. Interactions between closer individuals then causes clustering. However, the gap between these two modes around $r = 0.5$ means that two clusters form at this distance apart: too close for long range interactions and too far for short range ones. Recall that Lemma \ref{Lemma: MotschTadmor} requires either the interaction function or the distance between points to become arbitrarily small at some point in time. The two clusters lying a distance apart at which $\phi_3$ is zero satisfies Lemma \ref{Lemma: MotschTadmor} without then leading to consensus. While the general behaviour of the model is clearly the formation of opinion clusters, the free choice of interaction function creates great diversity in where, how many and how these clusters form. 

It is interesting to note that replacing the interaction function with another that is higher at every point does not always yield greater agreement. Intuitively, increasing the interaction function increases the amount of communication between individuals, leading to more agreement over time. Yet while this is often the case it is not necessarily true. Increasing the interaction function can also lead to faster local clustering, potentially preventing those interactions between clusters that would be required to achieve a global consensus. Increasing edge weights can have a similar effect. Specific examples are given in Appendix \ref{Appendix: Additional examples} in Examples \ref{example: increasing phi} and \ref{example: increasing w}. 

Another interesting feature of the fixed network model is the emergence of `isolated individuals'. Many common interaction functions reduce interaction strength as the distance between individuals' opinions grows, hence an individual with a small number of neighbours in $w$ and a relatively extreme opinion may have few or no interactions. Considering in particular those interaction functions where the interaction strength is zero beyond a critical value, some individuals may become totally isolated from the rest of the population. While it is possible for these individuals to interact with the rest of the network in the future, during the time in which they are isolated they will not cause any other individual's opinion to be moved towards their own. In addition, the opinion of an isolated individual will not move closer to that of the rest of the population. Hence, as is seen for $\phi_1$ in Fig.\ref{fig:fixed network example interaction functions}, isolated individuals typically remain isolated. This also causes their (often extreme) opinions to be underrepresented in the opinion formation process. The introduction of a dynamic network offers one solution to this problem. 

\section{Dynamic network} \label{Section: Dynamic network} 

In a real social network the relationships between individuals, and the extents to which they trust each others' opinions, are not static: new connections might occur either randomly or through mutual acquaintances and old connections may decay. To capture these effects we propose an extension to the fixed network model that couples the $N$ ODEs for opinion formation $(x_i)$ with $N^2$ ODEs for the edge weights $(w_{ij})$. 

As discussed previously, we interpret the interaction function as describing the extent to which individuals $i$ and $j$ interact at a given time based on the difference in their opinions, and interpret the edge weight as the strength of their relationship. It is not necessary that at a given moment in time these two factors should be related. However, new relationships between individuals do not immediately carry a strong weight, but instead build up this weight through continued interaction. Moreover, existing relationships decay if individuals do not continue to interact, or only interact weakly. We therefore propose dynamics of $w_{ij}$, that is the changes in the relationship between individuals $i$ and $j$, that are driven by their interactions. 

To model this we define two functions $f^+(w), f^-(w):[0,1]^{N \times N} \rightarrow \mathds{R}^{N \times N}$. The function $f^+$ describes how weights may increase and $f^-$ describes how they may decrease. The overall change in $w_{ij}$ is a weighting of $f^+(w)_{ij}$ and $-f^-(w)_{ij}$ (where $f(w)_{ij}$ denotes the $ij^\text{th}$ components of the function). We place several assumptions on these functions:

\begin{assumption} \label{Assumption group: f}
    The functions $f^+:[0,1]^{N \times N}\rightarrow \mathbb{R}^{N \times N}$ and $f^-\colon [0,1]^{N \times N} \rightarrow \mathbb{R}^{N \times N}$ are both non-negative and continuously differentiable. Furthermore they satisfy
    \begin{enumerate}[label=\alph*)]
        \item For any $i,j\in\{1,\dots,N\}$ and $w\in[0,1]^{N\times N}$, $w_{ij} = 1 \Rightarrow f^+(w)_{ij} \leq 0 $. 
        \item For any $i,j\in\{1,\dots,N\}$ and $w\in[0,1]^{N\times N}$, $w_{ij} = 0 \Rightarrow f^-(w)_{ij} \geq 0 $. 
    \end{enumerate}
\end{assumption}
These last two assumptions will ensure weights remain in the interval $[0,1]$. Regarding initial conditions, we assume that $w_{ij}(0)\in[0,1]$ for $i\neq j$ and $w_{ii}(0)=1$ for all $i\in\{1,\dots,N\}$. Then the opinion formation on a dynamic network reads as
\begin{subequations} \label{Eqn: general dynamic network system}
\begin{align}
    \frac{dx_i}{dt} &= \frac{1}{k_i(t)} \sum_{j\neq i} w_{ij}\,\phi\big(|x_i - x_j|\big)\,(x_j - x_i) \,, & i = 1,\dots,N \label{Eqn: general dynamic network system x_i}\\
 %   k_i &= \sum_{j=1}^N w_{ij} \,,& i = 1,\dots,N \label{Eqn: general dynamic network system k_i}\\
    \frac{dw_{ij}}{dt} &= \phi\big(|x_i - x_j|\big)\, f^+(w)_{ij} - \Big(1 - \phi\big(|x_i - x_j|\big)\Big)\, f^-(w)_{ij} \,, & i,j = 1,\dots,N,\, i\neq j \label{Eqn: general dynamic network system w_ij} \\[0.5em]
    \frac{dw_{ii}}{dt} &= 0 & i = 1,\dots,N.
\end{align}
\end{subequations}
We refer to this as the dynamic network model. Note that if both $f^+$ and $f^-$ are identically equal to zero, this gives the fixed network model. 
\begin{remark}
We recall  that $k_i(t) = \sum_{j=1}^N w_{ij}$. 
    In the dynamic network model this prefactor is dynamic. Furthermore, as the derivative of $w_{ii}$ is zero, $w_{ii}(t) = w_{ii}(0) = 1$ for all $t\geq0$, hence $k_i(t) \geq 1 > 0$ for all $t\geq0$. This ensures that the division by $k_i(t)$ is always defined. 
\end{remark}

\subsection{General properties} \label{Section: General properties}

As with the fixed network model, we first show that the dynamic network model is well-defined in the context of opinion dynamics. 

\begin{restatable}{proposition}{boundedness}  \label{Prop: Boundedness}
    Let Assumption \ref{Assumption group: phi} on $\phi$ and Assumption \ref{Assumption group: f} on $f^+$ and $f^-$ be satisfied. For $x(0)\in[-1,1]^N$ and $w(0)\in[0,1]^{N\times N}$, let $x(t)$ and $w(t)$ be the solution to \eqref{Eqn: general dynamic network system}. Then $x(t)\in[-1,1]^N$ and $w(t)\in[0,1]^{N\times N}$ for all times $t\geq 0$.
\end{restatable}
\begin{proof}
    As noted above $w_{ii}(t) = 1 \in [0,1]$ for all $t\geq0$. Let $i\neq j$ and consider $w_{ij}$. If $w_{ij} = 0$ then by Assumption \ref{Assumption group: f}, $f^-(w)_{ij} \geq 0$ and 
    \begin{equation}
        \frac{dw_{ij}}{dt} \geq \phi\big(|x_i - x_j|\big)\, f^+(w)_{ij} \geq 0.
    \end{equation}
    Hence $w_{ij}$ cannot fall below zero. If $w_{ij} = 1$ then, by Assumption \ref{Assumption group: f}, $f^+(w)_{ij} \leq 0$ and 
    \begin{equation}
        \frac{dw_{ij}}{dt} \leq \Big(1 - \phi\big(|x_i - x_j|\big)\Big)\, f^-(w)_{ij} \leq 0.
    \end{equation}
    Hence $w_{ij}$ cannot exceed one. So, since $w_{ij}(0)\in[0,1]$ we have that $w_{ij}(t)\in[0,1]$ for all $t\geq 0$. The remainder of the proof proceeds as in Proposition \ref{Prop: Boundedness (fixed network)}.
\end{proof} 

We now give a condition under which the system can be split into multiple smaller, independent systems. 

\begin{restatable}{proposition}{Independence}  \label{Prop: Independence}
    Let $\phi$ satisfy Assumption \ref{Assumption group: phi} and assume also that $\phi$ is a decreasing function with $\phi(0)>0$. Let $\Phi$ be a second adjacency matrix defined by $\Phi_{ij} = \phi\big(|x_i(0) - x_j(0)| \big)$. If individuals $i$ and $j$ are in different connected components of the network with adjacency matrix $\Phi$, then $\phi\big(|x_i(t) - x_j(t)|\big) = 0$ for all $t\geq0$. That is, individuals $i$ and $j$ never interact. 
\end{restatable}
\begin{proof}

    Firstly note that as $\Phi$ is determined only by the interaction function $\phi$, it is independent of the social network $w$. Assume two individuals $i$ and $j$ are in different connected components of $\Phi$. Relabel the individuals in the population according to their initial opinion, so that $x_1(0) \leq x_2(0) \leq \dots \leq x_N(0)$. Let $i'$, $j'$ be the new indices of individuals $i$ and $j$ respectively. Without loss of generality assume that $i' < j'$. 
    
    For $n = 2,\dots,N-1$ let $d_n = |x_{n}(0) - x_{n-1}(0)| \geq 0$ and consider the values $\phi(d_n)$. If the values $\phi(d_n)$ were all strictly positive then all individuals would be in the same connected component of $\Phi$. As $i'$ and $j'$ are assumed to be in different connected components of $\Phi$ there must exist an $n$ with $\phi(d_n) = 0$ and $x_{i'}(0) \leq x_{n-1}(0) < x_n(0) \leq x_{j'}(0)$. Note that the strict inequality $x_{n-1}(0) < x_n(0)$ arises since $\phi(0)>0$ and $\phi(d_n)=0$, so $d_n = |x_{n}(0) - x_{n-1}(0)| > 0$. 
    
    This partitions the population into two groups: individuals $1,\dots,n-1$ with initial opinions in the interval $[-1,x_{n-1}]$, and individuals $n,\dots,N$ with initial opinions in the interval $[x_{n},1]$. Individual $i'$ lies in the first of these and $j'$ in the second. There is a gap between these groups of length $d_n > 0$. As $\phi$ is assumed to be a decreasing function, $\phi(r) = 0$ for any $r \geq d_n$, hence there are no interactions between individuals $1,\dots,n-1$ and individuals $n,\dots,N$. 
    
    Using a similar argument as in Proposition \ref{Prop: Boundedness (fixed network)}, the maximum opinion held amongst the individuals in each of these groups cannot increase, and the minimum opinion held within each group cannot decrease. Therefore the gap distance between these two groups will always be at least $d_n$, meaning the distance between individuals $i'$ and $j'$ will always be at least $d_n$, so $\phi(|x_i(t) - x_j(t)|) = 0$ for all $t\geq 0$.  
\end{proof}
This proposition shows that for some interaction functions, for example, the bounded confidence interaction function, the system can be separated into distinct independent components. Furthermore, if the system reaches a state in which $\Phi$ has different connected components then these will be independent from that time onward. 

A common property of opinion formation models is that individuals' opinions remain in their initial order \cite{hegselmann2002opinion}. However, it should be noted that this property does not always hold in  the dynamic network model. A specific example is shown in Example \ref{example: Order change}. 

Another property that is not necessarily preserved under these dynamics is the symmetry of $w$, meaning that even if $w(0)$ is undirected, $w(t)$ may be directed at some later time. Proposition \ref{Prop: Undirected} provides a condition that guarantees this switch from an undirected to a directed network does not occur. 
\begin{definition}
    We call a function $f:[0,1]^{N\times N} \rightarrow \mathds{R}^{N \times N}$ symmetry preserving if $f(w)_{ij} = f(w)_{ji}$ for all $i,j \in \lbrace 1, \ldots N \rbrace$ for any symmetric matrix $w\in [0,1]^{N\times N}$. 
\end{definition}
\begin{restatable}{proposition}{undirected} \label{Prop: Undirected} 
    Let $\phi$ satisfy Assumption \ref{Assumption group: phi}. Let $f^+$ and $f^-$ satisfy Assumption \ref{Assumption group: f}. In addition, assume both $f^+$ and $f^-$ are symmetry preserving. For a symmetric initial weight matrix $w(0)\in[0,1]^{N\times N}$ and an initial opinion vector $x(0)\in[-1,1]^N$, let $x(t)$ and $w(t)$ be the solution to \eqref{Eqn: general dynamic network system}. Then $w(t)$ is symmetric for any $t\geq 0$.
\end{restatable}
\begin{proof}
    For any symmetric $w$ we have $f^+(w)_{ij} = f^+(w)_{ji}$ and $f^-(w)_{ij} = f^-(w)_{ji}$ for all $i,j$. As $\phi$ is a function of the distance between points we also have that $\phi(|x_i - x_j|) = \phi(|x_j - x_i|)$, hence the derivative of $w_{ij}$ is equal to that of $w_{ji}$ and $w$ remains symmetric. 
\end{proof}
\begin{remark}
    Note that $f^+$ or $f^-$ failing to be symmetry preserving does not mean that the symmetry of $w(0)$ will necessarily be lost. For example, $f^+$ may fail to be symmetry preserving only for a single symmetric network. If this network cannot be reached from given initial conditions then the symmetry of $w$ would be preserved. 
\end{remark}

\subsection{Examples} \label{Section: examples}

We now introduce three example weight dynamics that have an increasingly significant impact on the opinion formation process. Note that all three examples are symmetry preserving. 

\textbf{Logistic weight dynamics:} The purpose of logistic weight dynamics is to study the effect of changing edge weights only, without creating any new edges. 
\begin{equation}
    f^+(w)_{ij} = w_{ij}(1 - w_{ij}) \,,\quad f^-(w)_{ij} = w_{ij}(1 - w_{ij}) \,.
\end{equation}
This gives
\begin{align}
    \label{eqn: logistic weight dynamics}
    \frac{dw_{ij}}{dt} &= \phi\big(|x_i - x_j|\big)\,w_{ij}\,(1-w_{ij}) - \big(1 - \phi\big(|x_i - x_j|\big)\big)\,w_{ij}\,(1-w_{ij}) \nonumber\\
    &= w_{ij}\,(1-w_{ij})\,(2\phi\big(|x_i - x_j|\big) - 1 )
\end{align}
If the weight of an edge is zero or one it will remain at this fixed point. If $w_{ij}\in(0,1)$ the edge weight will increase towards $1$ if $i$ and $j$ interact sufficiently strongly ($\phi\geq0.5$), and decrease towards zero otherwise. 

\textbf{Friend-of-a-friend weight dynamics:} 
Friend-of-a-friend (FOAF) weight dynamics aims to add a mechanism through which new edges could be added to the network.
\begin{equation}
    f^+(w)_{ij} = \big( w_{ij} + N^{-1} (w^2)_{ij} \big)\,(1-w_{ij}) \,,\quad f^-(w)_{ij} = w_{ij} \,.
\end{equation}
This gives
\begin{align}
    \label{eqn: FOAF weight dynamics}
    \frac{dw_{ij}}{dt} &= \phi\big(|x_i - x_j|\big)\,\big( w_{ij} + N^{-1} (w^2)_{ij} \big)\,(1-w_{ij}) - \big(1 - \phi\big(|x_i - x_j|\big)\big)\,w_{ij} 
\end{align}
\begin{figure}[H]
    \centering
    \includegraphics[width=\linewidth]{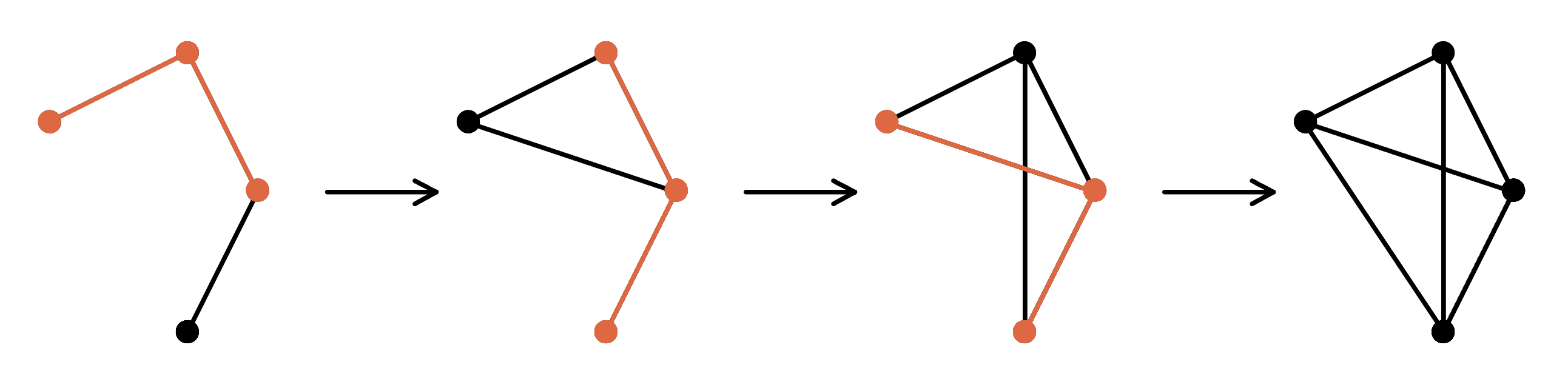}
    \caption{Successive triadic closures creating a complete network. In each step, the triad that will be closed is shown in orange.}
    \label{fig: triadic closure diagaram}
\end{figure}
The addition of the $(w^2)_{ij}$ term in $f^+_{ij}$ creates the possibility for an edge to form between individuals $i$ and $j$ if they are both connected to a `mutual friend', that is there exists some individual $\ell$ such that $w_{i\ell} w_{\ell j} > 0$. This allows the social network to grow through the process of `triadic closure', as shown in Fig.\ref{fig: triadic closure diagaram}. This is an established idea in both network and social sciences \cite{rapoport1953spread,heider1982psychology} and is thought to play an important role in the structure of social networks \cite{bianconi2014triadic,kossinets2009origins,klimek2013triadic,huang2018will}. Other opinion formation models (using instead a discretely evolving network) have also included updates to the network through triadic closure \cite{iniguez2009opinion}. Note that the $(w^2)_{ij}$ term in \eqref{eqn: FOAF weight dynamics} is scaled by a factor of $N^{-1}$ so that all terms are of the same order. 

\textbf{Memory weight dynamics:} The final example we consider may give the simplest model but has the most profound impact on the opinion formation process.
\begin{equation}
    f^+(w)_{ij} = (1-w_{ij}) \,,\quad f^-(w)_{ij} = w_{ij} \,.
\end{equation}
This gives
\begin{align}
    \label{eqn: memory weight dynamics}
    \frac{dw_{ij}}{dt} &= \phi\big(|x_i - x_j|\big)\,(1-w_{ij}) - \big(1 - \phi\big(|x_i - x_j|\big)\big)\,w_{ij} \nonumber \\
    &= \phi\big(|x_i - x_j|\big) - w_{ij} 
\end{align}
In the weight dynamics of \eqref{eqn: memory weight dynamics}, $w_{ij}$ simply moves towards $\phi\big(|x_i - x_j|\big)$, meaning that individuals $i$ and $j$ build and maintain a relationship through continually interacting. Moreover, \eqref{eqn: memory weight dynamics} can be obtained by differentiating 
\begin{equation} \label{eqn: weights as memory}
    w_{ij}(t) = e^{-t}\, w_{ij}(0) + \int_0^t e^{s-t}\,\phi\big(|x_i(s) - x_j(s)|\big) \,ds \,,
\end{equation}
hence $w_{ij}$ can be interpreted as giving a `memory' of how much individuals $i$ and $j$ have interacted in the past. More recent times are weighted more strongly, or equivalently more distant times become less significant. This mirrors a limited memory capacity for each individual, an idea also considered in \cite{stokes2022extremism} and \cite{mariano2020hybrid}. Here \eqref{eqn: weights as memory} can also be interpreted as pairs of individuals establishing a trusting relationship over time, if they continue to be in agreement. Conversely, individuals that continually disagree gradually lose confidence in each others' opinions. 

A variation on the memory weight dynamics is to include a limit on the strength of particular edges. This allows the dynamics to give individuals a `memory' while also preserving a network structure. Letting $m\in[0,1]^{N\times N}$ be a matrix of maximal edge weights define weight dynamics by
\begin{equation} \label{Eqn: adapted memory weight dynamics}
    \dfrac{dw_{ij}}{dt} = m_{ij} \phi\big(|x_i - x_j|\big) \,-\, w_{ij}.
\end{equation}
This limits which individuals can interact, perhaps due to spatial constraints. Similar adjustments can be made to the other examples considered above, although these are beyond the scope of this paper. 

\subsection{Case Study: Bounded confidence} \label{Section: BC}

\begin{figure}[ht!]
    \centering
    \includegraphics[width=0.8\linewidth, trim = {2.5cm 2.5cm 0.5cm 2.5cm}, clip]{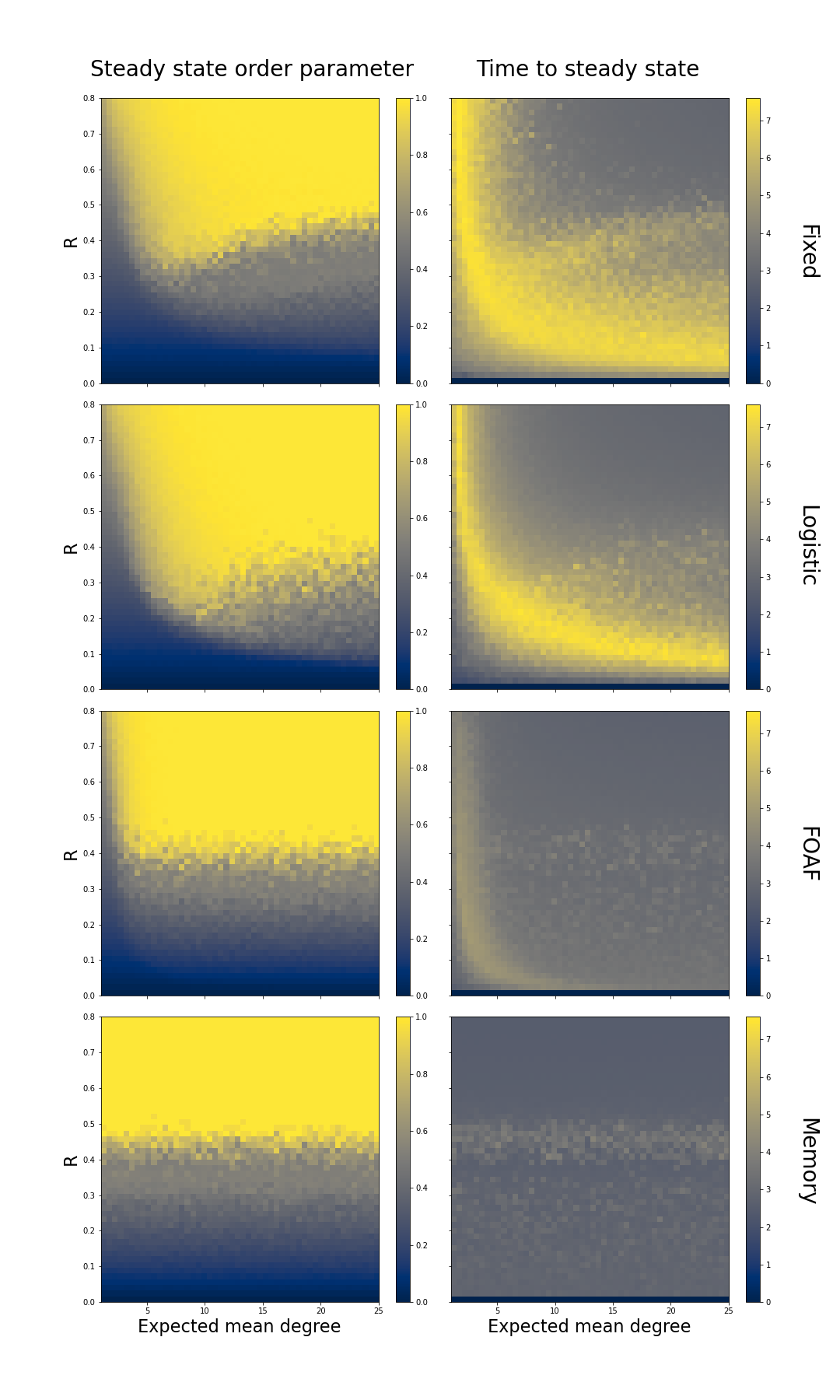}
    \caption{Heatmaps showing the order parameter at steady state (left column) and the time taken to reach that state (right column) for the dynamic network model \eqref{Eqn: general dynamic network system} with bounded confidence interaction function $\phi_R$ \eqref{Eqn: confidence indicator}. Each row of plots corresponds to a different type of weight dynamics. The expected mean degree of the initial network is varied along the x-axis of each heatmap and the bounded confidence radius $R$ is varied along the y-axis.}
    \label{fig:timed heatmaps, bounded confidence}
\end{figure}

To demonstrate the impact of dynamic weights we simulate the three examples given above, using two different interaction functions. The first set of results use the bounded confidence interaction function $\phi_R$ given in \eqref{Eqn: confidence indicator}. The confidence level $R$, as well as the mean degree of the initial network (denoted $\Bar{k}(0)$), are both varied. The dynamics are then simulated with fixed weights and each of the three example weight dynamics \eqref{eqn: logistic weight dynamics},\eqref{eqn: FOAF weight dynamics} and \eqref{eqn: memory weight dynamics}. For each set of simulations, the same initial opinions and initial network is used. A population size of $N = 500$ is used throughout. Initial opinions are chosen randomly from a uniform distribution on $[-1,1]$. Initial networks are Erd\H{o}s-R\'enyi random graphs with $N$ nodes and varying edge probabilities $p$. The expected value of $\Bar{k}(0)$ is given by $Np$. We have purposefully chosen small values of $\Bar{k}(0)$ so that the network plays a notable role in the opinion formation process. For each combination of $R$ value and $\Bar{k}(0)$, the dynamics are run for and results are averaged over 10 sets of initial conditions.

To quantitatively describe the outcome of our simulations we adapt an order parameter originally introduced for the HK model by Wang \textit{et al.} \cite{wang2017noisy}. Denoting by $M_\phi$ the maximum possible value of $\phi$, define the order parameter $Q$ as 
\begin{equation} \label{eqn: order parameter}
    Q = \frac{1}{N^2 M_{\phi}} \sum_{i,j=1}^N \phi\big(|x_i - x_j|\big) \,.
\end{equation}
The order parameter $Q$ takes values in the range $[0,1]$, with a value of $1$ indicating maximal interaction. For interaction functions that are decreasing, $Q$ also describes the extent of opinion clustering \cite{goddard2022noisy}. 

Simulations are run until a steady state is reached (defined as the difference in the population's opinions between consecutive timesteps of the discretisation being less than $10^{-5}$ in the Euclidean norm) or until a maximum time of $t=2000$. At this point, the order parameter is reported as well as the time required to reach this state. Results for the bounded confidence interaction function are shown in Fig.\ref{fig:timed heatmaps, bounded confidence}. Appendix \ref{Appendix: Difference heatmaps} also shows heatmaps of the difference between the results for the fixed network and each of the network dynamics. 

When using a fixed network we observe a transition from disagreement (appearing here as dark blue) in the bottom left corner to consensus (appearing here as yellow) in the top right corner. When $\Bar{k}(0)$ is relatively small ($<10$) its value has a significant effect on the opinion formation process. As the low mean degree restricts the potential for interaction, higher values of $R$ are needed in order to approach consensus. When $\Bar{k}(0)$ is relatively high ($>15$) it has less impact, as the population is already sufficiently well mixed. As is typical in bounded confidence models, higher values of $R$ lead to greater agreement \cite{hegselmann2002opinion,fortunato2005consensus}, although the introduction of the network appears to alter the location of the transition between clustering and consensus. It should be noted that the value of $R$  affects the calculation of the order parameter (as it is defined based on the interaction function) but $\Bar{k}(0)$ does not. 

The time taken to reach steady state is shown on the right hand side (on a log scale). It is clear that areas of the parameter space in which the behaviour of the model changes typically take much longer to reach their steady state. Consider for example a scenario in which the population is split between two distinct clusters, except for a small number of individuals that lie between the clusters. For small values of $R$ these clusters would remain separate and for large values of $R$ all individuals would interact and reach consensus. For intermediary values of $R$ we observe a situation in which the small group of individuals bridge the gap between the two clusters, slowly bringing them towards each other until they eventually interact and collapse into a single group. An example of this can be found in Appendix \ref{Appendix: Additional examples}. This behaviour is observed at the transition between polarisation and consensus and takes significantly longer to reach its steady state than either of the more extreme behaviours. 

The logistic weight dynamics \eqref{eqn: logistic weight dynamics} do not have a major effect on the opinion formation process, although there is slightly more agreement for high mean degrees. As logistic weight dynamics do not create any new edges, or remove any edges with weight $1$, it is perhaps unsurprising that their effect is minimal. A similar pattern is also seen in the time to steady state, with the logistic weight dynamics creating a slight speed-up in some areas. 

FOAF weight dynamics \eqref{eqn: FOAF weight dynamics} have a much greater impact. The effect of $\Bar{k}(0)$ is almost entirely removed as new edges are created. However, when $\Bar{k}(0)$ is extremely small it still has some effect, as the FOAF weight dynamics cannot create edges between different connected components of the network. The time to steady state is drastically reduced for all parameter values, although the areas of transition in the opinion formation process still take relatively longer. 

The memory weight dynamics \eqref{eqn: memory weight dynamics} completely eliminate the effect of the initial mean degree. The typical bounded confidence transition from disagreement to clustering to consensus is observed across all mean degree values. Similarly to the FOAF weight dynamics the time to steady state is significantly reduced. 

It is interesting to note that in both the fixed network and logistic dynamic network, increasing $\Bar{k}(0)$ does not always have a monotonic effect on agreement. In the region $R\in[0.3,0.4]$ we see that when $\Bar{k}(0)$ is extremely small ($\leq 5$) the final order parameter is very low, similarly it is low when $\Bar{k}(0)$ is large ($\geq 15$). However, in between these two areas there is a region of consensus. While the cause of this effect is not immediately apparent, it may be that this level of connectivity in the network is sufficient to facilitate interaction but not so great as to lead to rapid local clustering. This is a similar behaviour to that discussed in Example \ref{example: increasing w}. Notably, this effect is not seen in the FOAF and memory dynamics as many new edges are created and this `ideal' level of connectivity is lost. Indeed, in this parameter region, the level of agreement is lower under memory and FOAF weight dynamics, compared to the fixed and logistic dynamics models. The existence and location of this region may depend on the finite population size. 

This case study demonstrates the impact dynamic weights can have on opinion formation. In some cases we observe greater consensus, both in the overall level of agreement and the time taken to reach it. However, in certain areas the change to the mean degree of the network leads to increased local clustering instead of consensus. 

\subsection{Case Study: Exponential interaction function} \label{Section: exp}

The previous simulation setup is now repeated for a second interaction function. For $\alpha \geq 0$ we define
\begin{equation} \label{eqn: exp phi}
    \phi^\alpha(r) = e^{-\alpha r}. 
\end{equation}
This interaction function (with $\alpha = 6$) is one of the examples shown in Fig.\ref{fig:fixed network example interaction functions}. Here the parameter $\alpha$ is varied, as the parameter $R$ was for the bounded confidence interaction function. Note that for a fixed distance between opinions, higher values of $\alpha$ give lower interaction strengths. Although this interaction function is always strictly positive, unlike the bounded confidence interaction function, multiple interesting behaviours are still observed. 

Note that if the random network used in a particular simulation is connected, as in the example in Fig.\ref{fig:fixed network example interaction functions}, Proposition \ref{Prop: Fixed network convergence} guarantees consensus in the fixed network model with interaction function $\phi^\alpha$. Hence any deviation from convergence in the fixed network model can only be caused by the network containing multiple distinct connected components. 

We also show that for the memory weight dynamics consensus can be guaranteed for any initial network. In fact, we give a more general statement in Proposition \ref{Prop: general phi memory consensus}, which covers the case of $\phi^\alpha$. This is done by applying Theorem 2.3 from \cite{motsch2014heterophilious} which provides a sufficient condition for consensus. 

\begin{restatable}{proposition}{general phi, memory weights consensus} \label{Prop: general phi memory consensus}
    Let $\phi$ satisfy Assumption \ref{Assumption group: phi}. In addition, assume that there exists a constant $c > 0$ such that $\phi(r) > c$ for all $r\in[0,2]$. Then the solution to \eqref{Eqn: general dynamic network system} with weight dynamics given by \eqref{eqn: memory weight dynamics} reaches consensus for any $x(0)\in[-1,1]^N$ and $w(0)\in[0,1]^{N\times N}$. 
\end{restatable}
\begin{proof}
Consider the distance between the minimum value of $\phi(r)$ and each current edge weight:
\begin{equation*}
    \frac{d}{dt}(c - w_{ij}) = w_{ij} - \phi\big(|x_i - x_j|\big) \leq - (c - w_{ij}) \,.
\end{equation*}
Hence by Gronwall's inequality $c - w_{ij} \leq e^{-t}\,\big(c - w_{ij}(0)\big)$ or equivalently $w_{ij} \geq \beta(t)$, where
\begin{equation*}
    \beta(t) := c - e^{-t}\,\big(c - w_{ij}(0) \big) \,. 
\end{equation*}
If $w_{ij}(0) > c$ then $\beta(t)\geq c \geq c/2$ for all $t\geq 0$. Conversely, if $w_{ij}(0) \leq c$ then $\beta(t)$ is an increasing function of time and so $\beta(\log(2)) \geq c/2$ implies that $\beta(t) \geq c/2$ for all $t\geq \log(2)$. This guarantees that for all $i,j\in\{1,\dots,N\}$, $w_{ij}(t) \geq c/2$ for all $t \geq \log(2)$. 

Define
\begin{align*}
    a_{ij} &:= \frac{1}{k_i} w_{ij}\,\phi\big(|x_i - x_j|\big) \geq \frac{c^2}{2N} \,, \quad \text{for } j\neq i \\
    a_{ii} &:= 1 - \sum_{j\neq i} a_{ij} \geq \frac{1}{N} \geq \frac{c^2}{2N} \,,
\end{align*}
with the inequalities holding for (at least) all times $t \geq \log(2)$. Hence after this time all individuals are interacting. Considering then the system restarted at time $t = \log(2)$, Theorem 2.3 from \cite{motsch2014heterophilious} can be applied to guarantee consensus. 
\end{proof}

Proposition \ref{Prop: general phi memory consensus} holds for functions $\phi^\alpha$ \eqref{eqn: exp phi} with $c = e^{-2\alpha}$, therefore the case of memory weight dynamics has only one possible outcome: consensus. However, other examples of weight dynamics are not so straightforward. In Section \ref{Section: BC} the use of an indicator function as the interaction function meant that the derivative of $w_{ij}$ was given either entirely by $f^+(w)_{ij}$ or entirely by $f^-(w)_{ij}$. However, as $\phi^\alpha$ takes values outside of $\{0,1\}$, the derivative of $w_{ij}$ is a balance of both $f^+(w)_{ij}$ and $f^-(w)_{ij}$. This means it is possible for edge weights to decrease even though $\phi^\alpha$ is strictly positive. The contextual interpretation of this is that a weak interaction is not enough to maintain a relationship. For example, logistic weight dynamics have the form
\begin{equation*}
    \frac{dw_{ij}}{dt} = w_{ij}\,(1-w_{ij})\,(2\phi^\alpha \big(|x_i - x_j|\big) - 1 ) \,,
\end{equation*}
so $w_{ij}$ will be decreasing if $\phi^\alpha \big(|x_i - x_j|\big) < \frac{1}{2}$, or equivalently if $|x_i - x_j| > \frac{1}{\alpha} \log(2)$. Hence it is still possible for edge weights to decay and to observe outcomes other than consensus.  

\begin{figure}[ht!]
    \centering
    \includegraphics[width=0.8\linewidth, trim = {2.5cm 1cm 0.5cm 1.5cm}, clip]{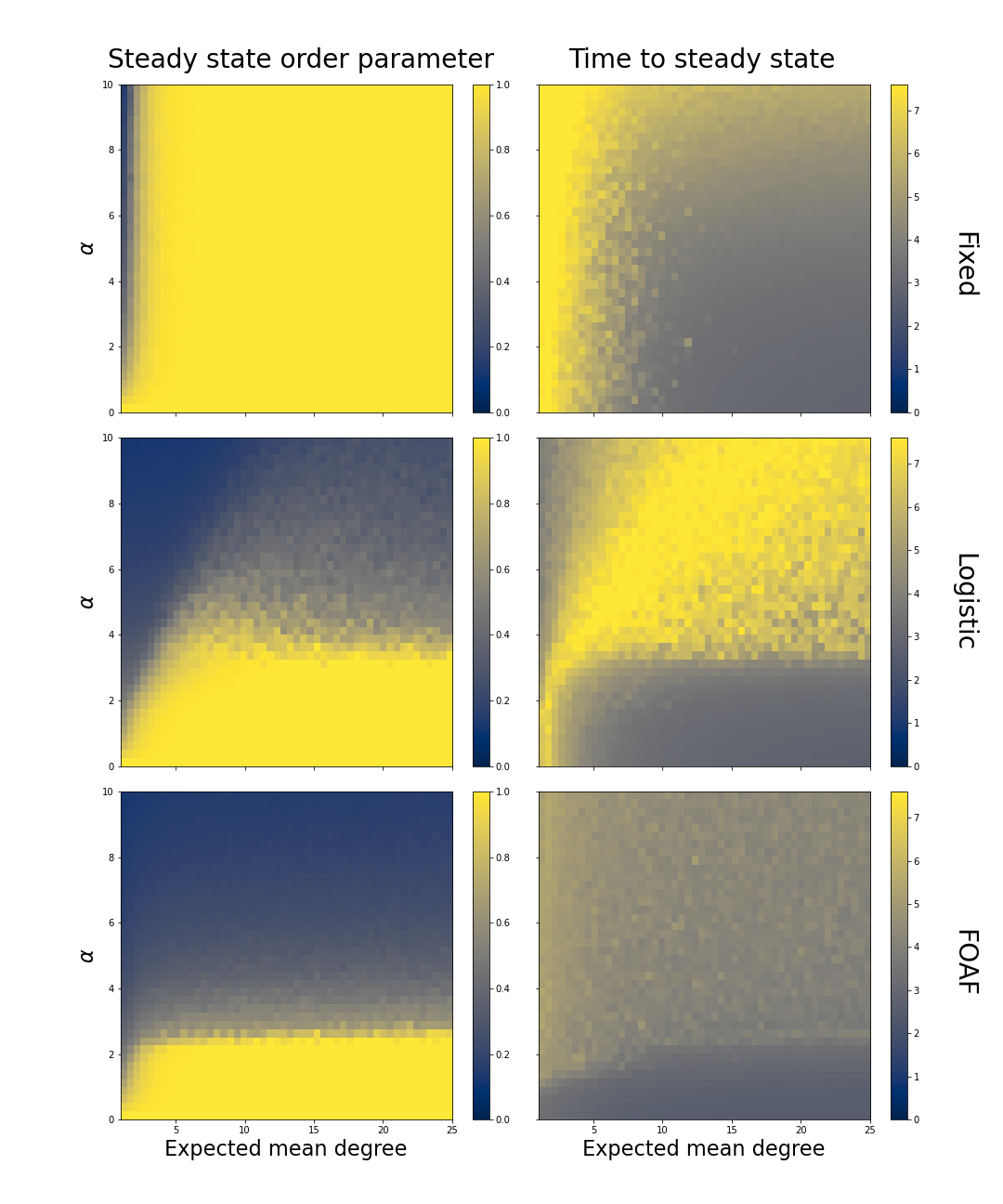}
    \caption{Heatmaps showing the order parameter at steady state (left column) and the time taken to reach that state (right column) for the dynamic network model \eqref{Eqn: general dynamic network system} with exponential interaction function $\phi^\alpha$ \eqref{eqn: exp phi}. Each row of plots corresponds to a different type of weight dynamics. The expected mean degree of the initial network is varied along the x-axis of each heatmap and the exponential decay rate $\alpha$ is varied along the y-axis.}
    \label{fig:timed heatmaps, exponential}
\end{figure}

Fig.\ref{fig:timed heatmaps, exponential} shows the outcome of the opinion formation process under each of: a fixed network, logistic weight dynamics, and FOAF weight dynamics. Recall that Proposition \ref{Prop: general phi memory consensus} guarantees consensus under memory weight dynamics. As previously Appendix \ref{Appendix: Difference heatmaps} shows heatmaps of the difference between the results for the fixed network and each of the network dynamics. Again note that $\alpha$ (but not the mean degree) affects the calculation of the order parameter. 

As expected the fixed network model \eqref{Eqn: fixed network ODEs} almost always reaches consensus, except in cases where $\Bar{k}(0)$ is extremely small and the fixed network has multiple connected components. This effect is in fact observed in all three cases. A small initial mean degree severely limits interactions, making reaching consensus more difficult. When new edges cannot be created between different connected components of $w$ consensus is often impossible if $w(0)$ is not connected. 

For logistic weight dynamics \eqref{eqn: logistic weight dynamics} we observe a transition from consensus at small values of $\alpha$, to clustering for larger values. For all values of $\alpha$, increasing the mean degree brings the population closer to consensus. 

As before, the time to steady state is much greater in areas of behavioural transition, with multiple parameter values reaching the maximum time (set here to $t = 3000$). For FOAF weight dynamics the time to steady state is much smaller and there is a clearer transition between consensus for small values of $\alpha$ and clustering for larger values of $\alpha$. As new edges can be created the initial mean degree has little impact. 

In this case study we observe logistic and FOAF weight dynamics both reducing the final level of agreement for many parameter values, demonstrating that weight dynamics can prevent as well as create consensus. 

\section{Extreme timescales} \label{Section: Timescales}

In this section we explore the impact of the relative timescales of the weight dynamics and opinion formation process, as done in \cite{barre2021fast}. This is done by considering two extreme scenarios: where the weight dynamics are significantly slower than the opinion formation process and where the weight dynamics are significantly faster. In both cases we investigate the effect of scaling the derivative of $w$, firstly by $\varepsilon > 0$ and secondly by $\varepsilon^{-1}$, then taking the limit $\varepsilon \rightarrow 0$. 

As this section compares the behaviour of the dynamic weight model to other models, we restrict our consideration to interaction functions $\phi$ that are Lipschitz continuous. In cases where $\phi$ is discontinuous, slight differences between opinion states can rapidly lead to completely different dynamics, preventing the type of comparison we wish to investigate. Hence we add the following assumption:
\begin{assumption} \label{assumption: phi Lipschitz} 
The interaction function $\phi$ is Lipschitz continuous with Lipschitz constant $L$.
\end{assumption}

Our first scenario of interest is when the weight dynamics are significantly slower than the opinion formation process. Intuitively, as the speed of the weight dynamics approaches zero the derivative of $w_{ij}$ also approaches zero, meaning the dynamics are similar to the fixed network model. 

\begin{theorem} \label{Theorem: slow weight dynamics}
    Let $\phi$ satisfy Assumption \ref{Assumption group: phi} and Assumption \ref{assumption: phi Lipschitz}. Let $f^+$ and $f^-$ satisfy Assumption \ref{Assumption group: f}. For $\varepsilon>0$ consider the system given by scaling the weight dynamics in \eqref{Eqn: general dynamic network system} by a factor of $\varepsilon$.
    \begin{subequations} \label{Eqn: slow network dynamics system}
    \begin{align}
        \frac{dx_i}{dt} &= \frac{1}{k_i(t)} \sum_{j=1}^N w_{ij}\,\phi\big(|x_i - x_j|\big)\,(x_j - x_i) \,, \\
        \frac{dw_{ij}}{dt} &= \varepsilon \Big( \phi\big(|x_i - x_j|\big)\, f^+_{ij}(w_{ij}) + \big(1 - \phi\big(|x_i - x_j|\big)\big)\, f^-_{ij}(w_{ij}) \Big).
    \end{align}
    \end{subequations}
    with initial conditions $x_i(0) = x_i^0\in[-1,1]$ and $w_{ij}(0) = w_{ij}^0\in[0,1]$. Define a second system in which weights are constant 
    \begin{align}
        \frac{dX_i}{dt} 
        &= \frac{1}{k_i^0} \sum_{j=1}^N w_{ij}^0\,\phi\big(|X_i - X_j|\big)\,(X_j - X_i) \,,
    \end{align}
    with matching initial condition $ X_i(0) = x_i^0 $. 
    Then there exist constants $C,\kappa>0$, independent of $t$ and $\varepsilon$, such that for all $t\geq0$, 
    \begin{equation}
        \sum_{i=1}^N | X_i(t) - x_i(t) | \leq \varepsilon \,C t^2 e^{\kappa t}
    \end{equation}
\end{theorem}
\begin{proof}
    As $f^+$ and $f^-$ are both continuous they are also bounded. Let $M>0$ be a constant such that $|f^+(w)_{ij}|<M$ and $|f^-(w)_{ij}|<M$ for all $w\in[0,1]^{N\times N}$ and $i,j\in\{1,\dots,N\}$. As $\phi\big(|x_i - x_j|\big)\in[0,1]$ this gives that $|w_{ij}(s) - w_{ij}^0| \leq 2Ms\varepsilon$ for all $s\geq 0$. 
    
    Let $t\geq 0$. Define $z_i(t) =  X_i(t) - x_i(t) $ and $u(t) := \sum_{i=1}^N | z_i(t) | $. Then we have
    \begin{align*}
        u(t) 
        &:= \sum_{i=1}^N | z_i(t) | \\
        &= \sum_{i=1}^N \bigg| \int_0^t \sum_{j=1}^N \bigg( \frac{w_{ij}^0}{k_i^0} \,\phi\big(|X_j(s) - X_i(s)|\big) \,(X_j(s) - X_i(s))  -  \frac{w_{ij}(s)}{k_i(s)} \,\phi\big(|x_j(s) - x_i(s)|\big) \,(x_j(s) - x_i(s)) \bigg) \,ds\, \bigg| \\
        %&\leq \sum_{i=1}^N \bigg| \int_0^t \sum_{j=1}^N \frac{w_{ij}^0}{k_i^0}\,\phi\big(|X_j(s) - X_i(s)|\big)\,(X_j(s) - X_i(s)) - \frac{w_{ij}^0}{k_i^0}\,\phi\big(|X_j(s) - X_i(s)|\big)\,(x_j(s) - x_i(s)) \,ds\, \bigg| \\ 
        %&\quad\,\, + \sum_{i=1}^N \bigg| \int_0^t \frac{w_{ij}^0}{k_i^0}\,\phi\big(|X_j(s) - X_i(s)|\big)\,(x_j(s) - x_i(s)) - \frac{w_{ij}(s)}{k_i(s)}\,\phi\big(|x_j(s) - x_i(s)|\big)\,(x_j(s) - x_i(s))\,ds\, \bigg| \\
        &\leq \sum_{i=1}^N \bigg| \int_0^t \sum_{j=1}^N \frac{w_{ij}^0}{k_i^0} \,\phi\big(|X_j(s) - X_i(s)|\big)\,(z_j(s) - z_i(s)) \,ds\, \bigg| \\
        &\quad\,\,+
        \sum_{i=1}^N \bigg| \int_0^t \sum_{j=1}^N
        \bigg( \frac{w_{ij}^0}{k_i^0}\,\phi\big(|X_j(s) - X_i(s)|\big) - \frac{w_{ij}(s)}{k_i(s)}\,\phi\big(|x_j(s) - x_i(s)|\big) \bigg)\,(x_j(s) - x_i(s))\,ds\, \bigg| \\
        &\leq \sum_{i=1}^N \sum_{j=1}^N \int_0^t \bigg|   \frac{w_{ij}^0}{k_i^0} \,\phi\big(|X_j(s) - X_i(s)|\big)\,\bigg| \big |z_j(s) - z_i(s)\big|\,ds\,\\
        &\quad\,\,+
        \sum_{i=1}^N \sum_{j=1}^N \int_0^t \bigg| 
        \frac{w_{ij}^0}{k_i^0}\,\phi\big(|X_j(s) - X_i(s)|\big) - \frac{w_{ij}(s)}{k_i(s)}\,\phi\big(|x_j(s) - x_i(s)|\big) \bigg|\, \big| x_j(s) - x_i(s) \big| \,ds.\,
    \end{align*}
    The first of these terms is relatively straightforward to simplify as both $w$ and $\phi$ are bounded above by $1$ and $k_i$ is bounded below by $1$. The second term is more complex, so we first perform the following calculation separately 
    \begin{align*}
         \bigg| \frac{w_{ij}^0}{k_i^0}\,\phi\big(|X_j(s) -& X_i(s)|\big) - \frac{w_{ij}(s)}{k_i(s)}\,\phi\big(|x_j(s) - x_i(s)|\big) \bigg| \\[0.2em]
        \leq& \bigg| \frac{w_{ij}^0}{k_i^0}\bigg( \phi\big(|X_j(s) - X_i(s)|\big) - \phi\big(|x_j(s) - x_i(s)|\big)\bigg) \bigg|  + \bigg| \bigg( \frac{w_{ij}^0}{k_i^0} - \frac{w_{ij}(s)}{k_i(s)}\bigg) \phi\big(|x_j(s) - x_i(s)|\big) \bigg| \\[0.4em]
        \leq& \bigg| \phi\big(|X_j(s) - X_i(s)|\big) - \phi\big(|x_j(s) - x_i(s)|\big) \bigg|  + \bigg| \frac{w_{ij}^0}{k_i^0} - \frac{w_{ij}(s)}{k_i(s)} \bigg| \\[0.4em]
        \leq& L\,\bigg| |X_j(s) - X_i(s)| - |x_j(s) - x_i(s)| \bigg|  + \bigg| \frac{w_{ij}^0\, | k_i(s) - k_i^0| +  k_i^0\,| w_{ij}^0 - w_{ij}(s)| }{k_i^0\,k_i(s)} \bigg| \\[0.4em]
        \leq& L\,\big| X_j(s) - X_i(s) - x_j(s) + x_i(s) \big|  + \frac{w_{ij}^0\, 2NMs\varepsilon +  k_i^0\,2Ms\varepsilon }{k_i^0\,k_i(s)} \\
        \leq& L\,\big| z_j(s) - z_i(s) \big|  + 4 N M\varepsilon s   
    \end{align*}
    Overall this gives
    \begin{align*}
        u(t) 
        &\leq \sum_{i=1}^N \sum_{j=1}^N \int_0^t \big |z_j(s) - z_i(s)\big|\,ds\,+
        2\sum_{i=1}^N \sum_{j=1}^N \int_0^t L \big |z_j(s) - z_i(s)\big| + 4 N M\varepsilon s  \, \,ds\, \\
        &\leq (1 + 2L) \sum_{i=1}^N \sum_{j=1}^N \int_0^t \,\big| z_j(s) \big| + \big|z_i(s) \big| \,ds\,+
        2\sum_{i=1}^N \sum_{j=1}^N \int_0^t 4 N M\varepsilon s  \, \,ds\, \\
        &\leq 2N(1 + 2L) \int_0^t u(s) \,ds\,+ 4N^3 M\varepsilon t^2.
    \end{align*}
    Letting $C = 4N^3 M$ and $\kappa = 2N(1+2L)$, Gronwall's inequality then gives the required result. 
\end{proof}

Hence the dynamic network model with slowed weight dynamics approximates the fixed network model arbitrarily well (for $\varepsilon$ sufficiently small). This confirms the intuition that if the weight dynamics are extremely slow the behaviour approaches that of the fixed weights model. 

\begin{remark}
    One approach to handling a discontinuous interaction function $\phi$ would be to mollify it, giving a continuous approximation that satisfies Assumption \ref{assumption: phi Lipschitz}. This would allow an application of the theorem above. However, as the mollified interaction function converges back to the original $\phi$, the Lipschitz constant of the mollified interaction function $(L)$ would necessarily diverge. Hence a balance would be needed between the accuracy of the mollification and the necessary size of $\varepsilon$. 
\end{remark}

We now discuss the scenario in which the weight dynamics are significantly faster than the opinion formation process. This scenario is more complex and there are significant restrictions on the types of weight dynamics for which a limit can be established. If the weight dynamics are relatively much faster than the opinion formation process, then the weights will rapidly approach their steady state during a time-frame in which individuals' opinions will have changed relatively little. In the limit as $\varepsilon\rightarrow0$ the weights can in fact be replaced by this steady state. The key issue is that this steady state must not depend on the initial weights, ensuring the behaviour of the weight dynamics is completely predictable. Of the three example weight dynamics we have considered only the memory weight dynamics has this property, hence for clarity we present Theorem \ref{theorem: fast memory weights} specifically for memory weight dynamics. Under additional assumptions the same approach could be applied to prove a similar result for other weight dynamics. 

\begin{theorem} \label{theorem: fast memory weights}
Let $\phi$ satisfy Assumption \ref{Assumption group: phi} and Assumption \ref{assumption: phi Lipschitz}. Furthermore assume $\phi(0)>0$. Let $X^0\in[-1,1]^N$. Define a new set of dynamics for $X =(X_1,\dots,X_N)$ by,  
\begin{subequations} \label{eqn: fast memory dynamics limiting model}
\begin{align}
    \frac{dX_i}{dt} &= \frac{1}{K_i} \sum_{j\neq i} \,\phi\big(|X_i - X_j|\big)^2\,(X_j - X_i),\quad X_i(0) = X_i^0 \\
    K_i &= \sum_{j=1}^N \phi\big(|X_i - X_j|\big). 
\end{align}
\end{subequations}
For $\varepsilon > 0$, consider the dynamic network model given by
\begin{subequations}
\begin{align}
    \frac{dx_i}{dt} &= \frac{1}{k_i} \sum_{j\neq i} w_{ij}\,\phi\big(|x_i - x_j|\big)\,(x_j - x_i) \,, \\ 
   % k_i &= \sum_{j=1}^N w_{ij} \,\nonumber\\
    \frac{dw_{ij}}{dt} &= \frac{1}{\varepsilon}\Big( \phi\big(|x_i - x_j|\big)\, 
 - w_{ij} \Big) \,, \label{eqn: dwdt memory fast}
 \end{align}
\end{subequations} 
with $x_i(0) = X_i^0$ and $w_{ij}(0) = \phi\big(|X_i^0 - X_j^0 |\big)$. Then there exist constants $K,c > 0$, independent of $t$ and $\epsilon$, such that 
\begin{equation} \label{eqn: fast memory dynamics result}
        \|x(t) - X(t)\|^2 \leq \varepsilon^2 c e^{Kt} \,.
\end{equation}
where $\|\cdot\|$ is the standard Eulcidean norm in $\mathds{R}^N$.
\end{theorem}
\begin{proof}
    This statement is a direct application of Theorem 15.2 from \cite{pavliotis2008multiscale}. The continuity of $\phi$ and linearity in $w$ of the memory weight dynamics are sufficient to ensure all conditions of Theorem 15.2 are met. 
\end{proof}

Hence in this case a `limiting model' exists as $\varepsilon\rightarrow0$, given by \eqref{eqn: fast memory dynamics limiting model}. Note that the interaction function appearing here is $\phi^2$, but the normalisation is by the sum of $\phi$'s. The additional assumption that $\phi(0) > 0$ ensures this normalisation is well-defined. The inclusion of $\phi$ in the normalisation makes this model distinct from any others considered in this paper. Furthermore, the appearance of both $\phi$ and $\phi^2$ makes this different from any standard opinion formation models. 

Appendix \ref{Appendix: Simulation of extreme timescales} shows an example numerical simulation of the extreme timescales described in Theorem \ref{Theorem: slow weight dynamics} and Theorem \ref{theorem: fast memory weights}. 

\section{Discussion} \label{Section: Discussion}

We have examined the behaviour of both the fixed and dynamic network models analytically in Sections \ref{Section: Fixed network} and \ref{Section: Dynamic network}, numerically through extensive simulation of two case studies in Sections \ref{Section: BC} and \ref{Section: exp}, and at extreme timescales in Section \ref{Section: Timescales}. Firstly we see that the behaviour of the dynamic network model can resemble that of the fixed network model, especially if the weight dynamics are slow. Many typical features of opinion dynamics models, such as a transition from consensus to polarisation, are observed in Fig.\ref{fig:timed heatmaps, bounded confidence} and Fig.\ref{fig:timed heatmaps, exponential}. However, we have also shown that weight dynamics can significantly alter the outcome of the opinion formation process. For example, Proposition \ref{Prop: general phi memory consensus} shows that memory weight dynamics in particular can guarantee consensus for a broad class of interaction functions. In addition, the bounded confidence case study in Section \ref{Section: BC} demonstrates that FOAF weight dynamics, which provide a more realistic mechanism for introducing new edges, can also help create consensus. On the other hand, the second case study in Section \ref{Section: exp} shows that these same weight dynamics can also drastically reduce the emergence of consensus by removing edges between individuals that interact only `weakly'. It is clear that different combinations of weight dynamics and interaction functions have varied and profound effects on opinion formation. 

Of course there are many other examples of weight dynamics that could have been tested and may produce interesting results. The question of which weight dynamics are most realistic, or which fit best to data, is of great interest. Furthermore, as mentioned in the description of \eqref{Eqn: adapted memory weight dynamics} it is also possible to modify the weight dynamics functions to impose a maximal edge weight for each pair of individuals. This enforces additional network structure that may have an interesting impact on examples such as memory weight dynamics. 

While our focus in this paper has primarily been on the impact of weight dynamics on opinion formation, there is of course an impact on the structure of the network itself. In the case of logistic weight dynamics, no new edges are introduced and the network can only become more sparse over time. As discussed, FOAF weight dynamics allows both for new edges to form and existing edges to decay, hence changing the overall structure of the network. One direction for future work is to more closely examine how the structure of the network changes, for example how the degree distribution evolves and the extent to which community structure emerges. Another significant development would be to change the population size during the dynamics. In this case a mechanism, such as preferential attachment \cite{barabasi1999emergence}, would be needed to establish a new individual's initial connections. The choice of this mechanism would itself have a significant impact on the network structure, and therefore on the opinion formation process. Considering a change in the population size also raises questions about the existence of a mean-field or large population limit, or indeed how this would be interpreted in the context of a network. The graph limit discussed in \cite{ayi2021mean} could provide one possible approach. 

In some cases, such as memory weight dynamics, the original network structure is entirely overridden. Theorem \ref{theorem: fast memory weights} shows that this can lead to the emergence of new models if this process happens extremely quickly. Examining the behaviour of these models is another direction for future research, as is an exploration of how the dynamic weights model behaves at intermediate timescales. Appendix \ref{Appendix: Simulation of extreme timescales} provides an example of how the dynamic network model changes with the relative timescale of the weight and opinion dynamics, moving from the fixed network model \eqref{Eqn: fixed network ODEs} to the standard dynamic network model \eqref{Eqn: general dynamic network system} to the fast network limit \eqref{eqn: fast memory dynamics result}. The behaviour of the model during this transition is not yet understood. 

It is also interesting to consider other applications in which this formulation of dynamic weights could be introduced. We focus here on the subject of opinion dynamics, which provides the general form of the dynamics for $x_i$. However, other collective behaviours such as flocking could also be treated similarly. 

Another natural extension to this work is to consider the addition of noise, arising either from external sources or uncertainty in measuring individuals' opinions, as in \cite{su2017noise} and \cite{goddard2022noisy}. It was shown in \cite{wang2017noisy} that in the SDE version of the HK model, including any amount of noise eventually leads to consensus due to the random movement of whole clusters. The introduction of a fixed or dynamic network could prevent consensus, leading to a broader array of long-time behaviours. However, adding noise to \eqref{Eqn: general dynamic network system} is not straightforward as special consideration is required for edges with weight zero, that is edges that are not present in the network. Simple additive noise would lead to the formation of edges between all pairs of individuals, removing the desired network structure. Furthermore, the choice of boundary conditions has been shown to have a major impact on opinion formation \cite{goddard2022noisy}. One approach to including random effects is to add and remove edges according to a Poisson process as in \cite{barre2021fast}, which also studied the limit as this process becomes extremely fast, although how such a process could be adapted to include the type of weight dynamics considered here is not clear. 

In summary, the dynamic weights model introduced in this paper provides a new approach to the topic of evolving networks that allows weights to change continuously in time in response to individuals' interactions. This can have a significant impact on the opinion formation process, both creating and preventing consensus when combined with different interaction functions. There are multiple possible directions for further research and we hope that in the future this approach could be developed to improve the accuracy of opinion forecasting and our understanding of how our opinions and social networks evolve together. 

\bibliography{bibliography.bib}
\bibliographystyle{unsrt}

\section*{Acknowledgements}

AN was supported by the Engineering and Physical Sciences Research Council through the Mathematics of Systems II Centre for Doctoral Training at the University of Warwick (reference EP/S022244/1). MTW is partly support by the Royal Society International Exchange Grant IES/R3/213113.

\appendix

\newpage
\section{Example Behaviours} \label{Appendix: Additional examples}

This appendix provides examples of several interesting phenomena observed in the fixed and dynamic network models. For clarity we typically consider a small population size of $N=3$, but all examples could be generalised to a larger population. Moreover, the behaviours displayed here are not unique to these specific examples, which merely serve as a simple demonstration of what it is possible to observe. 

Our first two examples concern the fixed network model \eqref{Eqn: fixed network ODEs} and show that increasing the amount of interaction between individuals can in fact lead to a loss of consensus. The first example focuses on the interaction function and the second on the (fixed) network. 

\begin{example} \label{example: increasing phi}
    We define the following two interaction functions, both of which satisfy Assumption \ref{Assumption group: phi}:
    \begin{align}
        \phi_1(r) &= \frac{1}{30} \mathds{1}\big\{r < 0.81 \big\} \,,\\
        \phi_2(r) &= \frac{1}{30} \mathds{1}\big\{ r < 0.81 \big\} + \frac{29}{30} \mathds{1}\big\{ r < 0.3 \big\} \,.
    \end{align}
    We also define initial conditions and a fixed network,
    \begin{equation}
        \begin{pmatrix}
            x_1(0) \\ x_2(0) \\ x_3(0)
        \end{pmatrix} 
        = \begin{pmatrix}
            -0.2 \\ 0 \\ 0.8
        \end{pmatrix} \,,\quad
        w = \begin{pmatrix}
            1 & 1 & 0 \\
            1 & 1 & 1 \\
            0 & 1 & 1
        \end{pmatrix} \,.
    \end{equation}
    As $\phi_2(r) > \phi_1(r)$ for all $r\in[0,2]$ it may be expected that individuals would interact more and hence end up with closer opinions.  In fact, the opposite occurs. Fig.\ref{fig:increasing interaction function} shows the dynamics using each interaction function, with $\phi_1$ on the left and $\phi_2$ on the right. Plots of $\phi_1$ and $\phi_2$ are shown in insets. Individual's trajectories are coloured according to their initial opinions. Note that time is shown on a log scale.
    
    Using interaction function $\phi_1$ the three individuals reach consensus. The higher value of $\phi_2(r)$ for $r\in[0,0.3]$ creates a significantly stronger interaction between individuals 1 and 2 that causes individual 2 (in the middle) to lower their opinion, ending their interaction with individual 3 and resulting in a loss of consensus. 
    
    \begin{figure}[ht!]
        \centering
        \includegraphics[width = \linewidth, trim = {2cm 0cm 3cm 1cm}, clip]{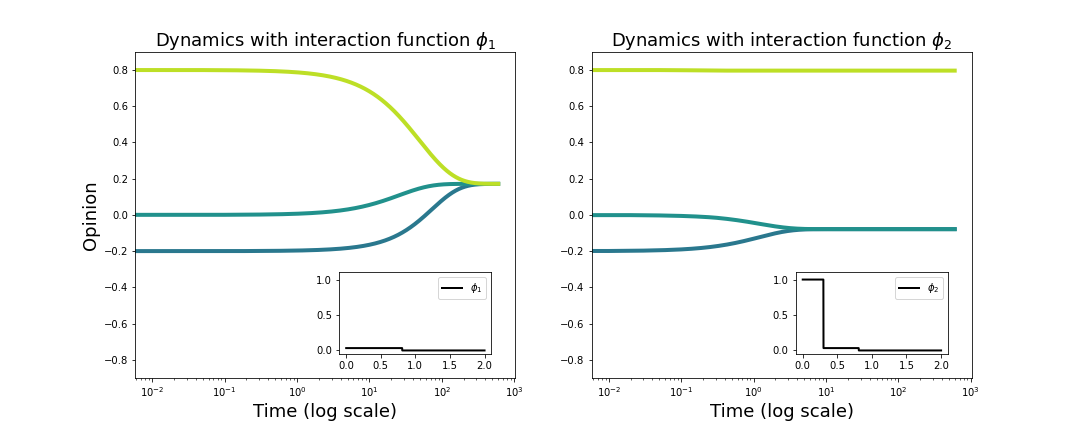}
        \caption{Example dynamics under different interaction functions (shown in inset plots). Individual's trajectories are coloured according to their initial opinion. This example shows that increasing the value of the interaction function does not always lead to greater agreement.}
        \label{fig:increasing interaction function}
    \end{figure}
\end{example}

\begin{example} \label{example: increasing w}
    In this example, we fix an interaction function $\phi(r) = \mathds{1}\big\{r < 0.8 \big\}$ and instead alter the network. Let
    \begin{equation}
        \begin{pmatrix}
            x_1(0) \\ x_2(0) \\ x_3(0)
        \end{pmatrix} 
        = \begin{pmatrix}
            -0.75 \\ 0 \\ 0.75
        \end{pmatrix} \,,\quad
        w^{(1)} = \begin{pmatrix}
            1 & 0.1 & 0 \\
            0.1 & 1 & 0.1 \\
            0 & 0.1 & 1
        \end{pmatrix} \,,\quad
        w^{(2)} = \begin{pmatrix}
            1 & 0.1 & 0 \\
            0.1 & 1 & 1 \\
            0 & 1 & 1
        \end{pmatrix} \,.
    \end{equation}
    Fig.\ref{fig:increasing weights} shows the dynamics using each network. Plots of the networks are shown in insets, with edge's colour indicating their weight. Individual's trajectories are again coloured according to their initial opinions and again time is shown on a log scale. As in the previous example we see a loss of consensus. This time the increased strength of edges $w_{23}$ and $w_{32}$ leads to a stronger interaction that brings $x_2$ away from $x_1$, ending their interaction. 
    
    \begin{figure}[ht!]
        \centering
        \includegraphics[width = \linewidth, trim = {2cm 0cm 3cm 1cm}, clip]{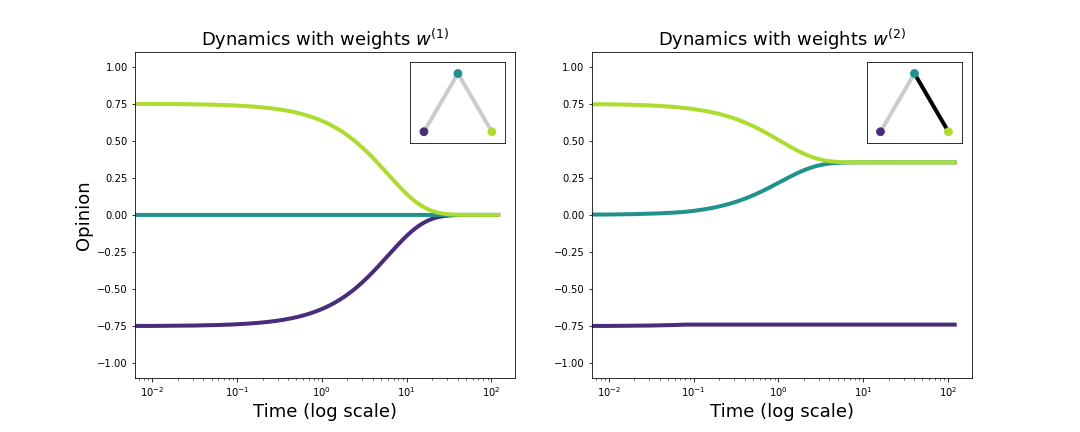}
        \caption{Example dynamics using different networks (shown in inset plots). Individual's trajectories are coloured according to their initial opinion. This example shows that increasing the strength of edges does not always lead to greater agreement.}
        \label{fig:increasing weights}
    \end{figure}
\end{example}

In both Example \ref{example: increasing phi} and Example \ref{example: increasing w}, consensus is lost due to an imbalance in interactions, leading to fast local clustering rather than a slower approach to consensus. 

Next we consider the dynamic network model \eqref{Eqn: general dynamic network system} and provide an example showing that the order of individuals' opinions is not necessarily preserved. Again a small population size of $N=3$ is chosen, but here a general interaction function $\phi$ and general weight dynamics $f^+$ and $f^-$ are used. 

\begin{example} \label{example: Order change} 

    Let $\phi$ be an interaction function satisfying Assumption \ref{Assumption group: phi}. Assume also that there exist constants $R,c > 0$ such that $\phi(r) > c$ for all $r<R$. Let $f^+$ and $f^-$ satisfy Assumption \ref{Assumption group: f}. 

    For some $K>2$, which will be determined later, set the initial conditions
    \begin{equation}
        x(0) = 
        \begin{pmatrix}
        -\frac{R}{K} \\ 0 \\ \frac{R}{2}
        \end{pmatrix}
        \,,\quad
        w(0) = 
        \begin{pmatrix}
            1 & 1 & 1 \\
            1 & 1 & 0 \\
            1 & 0 & 1 \\
        \end{pmatrix}.
    \end{equation}
    Note that $x_1(0)<x_2(0)<x_3(0)$. We show that, for $K$ sufficiently large, this initial ordering of opinions is not preserved. We assume that $x_1(t) \leq x_2(t)$ for all $t\geq 0$ and aim to reach a contradiction.
    
    By Proposition \ref{Prop: Boundedness} and the assumption that $K>2$, we have that $x_i(t)\in[-\frac{R}{K},\frac{R}{2}]\subset[-\frac{R}{2},\frac{R}{2}]$ for all $i=1,2,3$ and $t\geq 0$. Hence the distance between any pair of individuals is always strictly less than $R$ and so $\phi(|x_i(t) - x_j(t)|) > c$ for any individuals $i,j$ and $t\geq 0$. 

    $f^+$ and $f^-$ are both continuous functions, hence they are bounded. Let $M>0$ be a constant such that $|f^+(w)_{ij}|,|f^-(w)_{ij}| < M$ for any $w\in[0,1]^{N\times N}$ and $i,j\in\{1,\dots,N\}$. 

    We now derive bounds on the derivatives of $w$ and $x$, with the goal of obtaining a lower bound for the opinion $x_1(t)$ and an upper bound for the opinion $x_2(t)$. Beginning with entries in $w$:
    \begin{align*}
        \frac{dw_{12}}{dt} 
        %&= f^+(w)_{12}\,\phi(|x_1 - x_2|) - f^-(w)_{12}\,\big(1  - \phi(|x_1 - x_2|) \big) \,,\\
        &\geq - f^-(w)_{12}\,\big(1  - \phi(|x_1 - x_2|) \big) \,,\\
        %&\geq -M (1 - \varepsilon)\,, \\
        &\geq -M\,,
    \end{align*}
    hence $w_{12}(t) \geq 1 - Mt$. Similarly $w_{21}(t),w_{13}(t),w_{31}(t)\geq 1 - Mt$.
    \begin{align*}
        \frac{dw_{23}}{dt} 
        %&= f^+(w)_{23}\,\phi(|x_2 - x_3|) - f^-(w)_{23}\,\big(1  - \phi(|x_2 - x_3|) \big) \,,\\
        &\leq f^+(w)_{23}\,\phi(|x_2 - x_3|) \,\leq M \,,
    \end{align*}
    hence $w_{23}(t) \leq Mt$. Similarly $w_{32}(t)\leq Mt$. Now considering entries in $x$:
    \begin{align*}
        \frac{dx_1}{dt} 
        %&= \frac{1}{k_1} \sum_{j=1}^N w_{1j}\, \phi(|x_1 - x_j|)\, (x_j - x_1) \,,\\
        &= \frac{1}{k_1} \Big( w_{12}\, \phi(|x_1 - x_2|)\, (x_2 - x_1) + w_{13}\, \phi(|x_1 - x_3|)\, (x_3 - x_1) \Big) \,,\\
        &\geq \frac{1}{k_1} w_{13}\, \phi(|x_1 - x_3|)\, (x_3 - x_1) \,,\\
        &\geq \frac{1}{3} (1 - Mt)\, c\, (x_3 - x_1) \,.
    \end{align*}
    To proceed further requires a bound on $x_3 - x_1$. As all opinions remain within the interval $[-\frac{R}{2},\frac{R}{2}]$, for any $i$ we have that $|\frac{dx_i}{dt}|\leq NR = 3R$, so $x_3 - x_1 \geq \frac{R}{2} + \frac{R}{K} - 6Rt$. This gives
    \begin{align*}
        \frac{dx_1}{dt} &\geq  \frac{c}{3} (1 - Mt)\, \bigg(\frac{R}{2} + \frac{R}{K} - 6Rt \bigg) \,,\\
        &\geq  \frac{c R}{3}\, \Bigg(\frac{1}{2} + \frac{1}{K} - \bigg( \frac{M}{2} + \frac{M}{K} + 6 \bigg)t \Bigg) \,,
    \end{align*}
    hence 
    \begin{align*}
        x_1(t) 
        &\geq -\frac{R}{K} + \frac{c R}{3}\, \bigg(\frac{1}{2} + \frac{1}{K} \bigg)t - \frac{c R}{6}\, \bigg( \frac{M}{2} + \frac{M}{K} + 6 \bigg)t^2\,\\
        %&\geq -\frac{R}{K} + \frac{c R}{6}t - \frac{c R}{6}\, \bigg( \frac{M}{2} + \frac{M}{K} + 6 \bigg)t^2\,\\
        &\geq -\frac{R}{K} + \frac{c R}{6}t - \frac{c R}{6}\, \bigg( \frac{3M}{2} + 6 \bigg)t^2 \,.
    \end{align*}
    With the final inequality requiring $K \geq 1$. Moving onto $x_2$:
    \begin{align*}
        \frac{dx_2}{dt} 
        %&= \frac{1}{k_2} \sum_{j=1}^N w_{2j}\, \phi(|x_2 - x_j|)\, (x_j - x_2) \,,\\
        &= \frac{1}{k_2} \Big( w_{21}\, \phi(|x_2 - x_1|)\, (x_1 - x_2) + w_{23}\, \phi(|x_2 - x_3|)\, (x_3 - x_2) \Big) \,,\\
        &\leq \frac{1}{k_2} w_{23}\, \phi(|x_2 - x_3|)\, (x_3 - x_2) \,,\\
       %&\leq w_{23}\, \phi(|x_2 - x_3|)\, (x_3 - x_2) \,,\\
        &\leq Mt\, (x_3 - x_2) \,.
    \end{align*}
    As before we now require a bound on $x_3 - x_2$, however as we need an upper bound this case is simpler. As all opinions remain within the interval $[-\frac{R}{2},\frac{R}{2}]$, $x_3 - x_2 \leq R$. This gives,
    \begin{equation*}
        x_2(t) \leq \frac{MR}{2}t^2 \,.
    \end{equation*}
    Define a function of time $g(t)$ to be the difference between the lower bound for $x_1$ and the upper bound for $x_2$, with a factor of $R$ removed for convenience. 
    \begin{align*}
        g(t) &=  -\frac{1}{K} + \frac{c}{6}\,t - \frac{c}{6}\, \bigg( \frac{3M}{2} + 6 \bigg)t^2 - \frac{M}{2}t^2 \,,\\
        %&=  -\frac{1}{K} + \frac{c}{6}\,t -\, \bigg( \frac{Mc}{4} + c +  \frac{M}{2}\bigg)t^2 \,,\\
        &=  -\frac{1}{K} + \frac{c}{6}\,t -\, \frac{M(c + 2) + 4c}{4}\, t^2 \,,\\
        &=  -\frac{1}{K} + t\, \bigg( \frac{c}{6}\, -\, \frac{M(c + 2) + 4c}{4}\, t \bigg) \,.
    \end{align*}
    Letting
    \begin{equation*}
        \Bar{t} = \frac{c}{3M(c+2) + 12c} > 0,
    \end{equation*}
    we have that 
    \begin{equation}
        \Bar{t}\,\bigg( \frac{c}{6}\, -\, \frac{M(c + 2) + 4c}{4}\,\Bar{t} \bigg) = \Bar{t}\,\frac{c}{12} > 0 \,,
    \end{equation}
    so 
    \begin{equation}
        g(\Bar{t}) > \Bar{t}\,\frac{c}{12} - \frac{1}{K} \,.
    \end{equation}
    Hence choosing $K$ sufficiently large (and exceeding $1$) gives that 
    $g\big(\Bar{t}\big) > 0$, meaning $x_1\big(\Bar{t}\big) > x_2\big(\Bar{t}\big)$. This contradicts the assumption that $x_1(t) \leq x_2(t)$ for all $t\geq 0$, meaning that there must exist an $s\geq 0$ at which $x_1(s) > x_2(s)$ and therefore that the initial order of opinions is not preserved. 
    
    Note that as $\Bar{t}$ was derived under the (false) assumption that $x_1(t) \leq x_2(t)$ for all $t\geq 0$, it itself does not necessarily provide a time at which $x_1 > x_2$.
\end{example}

For our final example we return to the fixed network model \eqref{Eqn: fixed network ODEs}, motivated by the results of the case study in Section \ref{Section: BC}. In Fig.\ref{fig:timed heatmaps, bounded confidence} we observe that areas of the parameter space in which the behaviour of the model changes typically take much longer to reach their steady state. Example \ref{example: transition} demonstrates why this may be the case. 

\begin{example} \label{example: transition}
    For a population size of $N=201$ we evenly spaced initial opinions and set a regular network $w$. The fixed network model \eqref{Eqn: fixed network ODEs} is then simulated using the bounded confidence interaction function $\phi_R$ \eqref{Eqn: confidence indicator} for each of $R=0.35$, $R=0.4$ and $R=0.45$. In addition, the same condition as described in Section \ref{Section: BC} is used to determine when the steady state has been reached. This time is shown by a vertical dashed line. For $R=0.4$ and $R=0.45$ the population reaches consensus. However, it is clear that the middle value $R=0.4$ takes significantly longer to reach its steady state as it first forms several clusters that then later combine. This demonstrates why areas of parameter space in which the model changes behaviour often correspond to those which take longer to reach their steady state. 
    \begin{figure}[ht!]
        \centering
        \includegraphics[width = \linewidth]{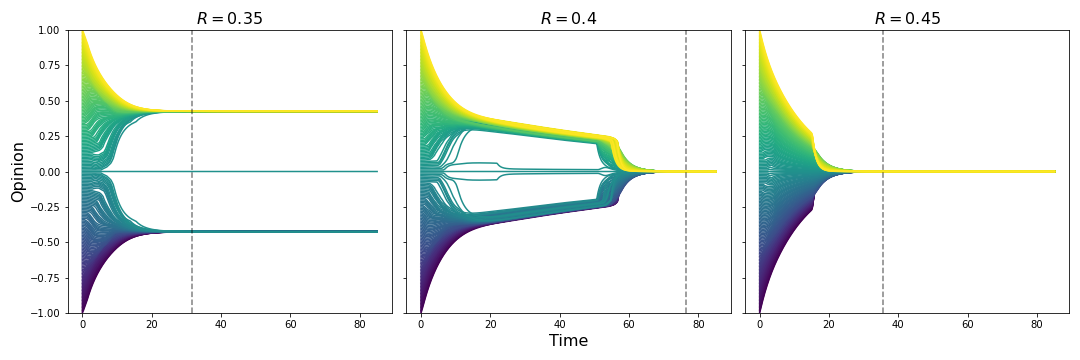}
        \caption{Example dynamics of the fixed network model \eqref{Eqn: fixed network ODEs} with bounded confidence interaction function $\phi_R$ \eqref{Eqn: confidence indicator} for three values of $R$. The time at which the steady state is reached is shown by a vertical dashed line. The middle panel, in which the behaviour transitions from clustering to consensus, takes the longest to reach its steady state.}
        \label{fig:transition example}
    \end{figure}
\end{example}

\clearpage
\section{Difference heatmaps} \label{Appendix: Difference heatmaps}

To help show the impact of the dynamic network we present below heatmaps of the difference in results between the fixed and dynamic network models. The results arise from the numerical simulations performed in Sections \ref{Section: BC} and \ref{Section: exp}. Blue indicates areas in which the dynamic network model had a higher order parameter than the fixed network model and red indicates the opposite. Grey shows areas in which the models produce the same final order parameter. Fig.\ref{fig:difference heatmap BC} shows results from the bounded confidence case study and Fig.\ref{fig:difference heatmap exp} shows results from the exponential interaction function case study.  

\begin{figure}[ht!]
    \centering
    \includegraphics[width = \linewidth, trim = {3cm 2cm 3cm 2cm}, clip]{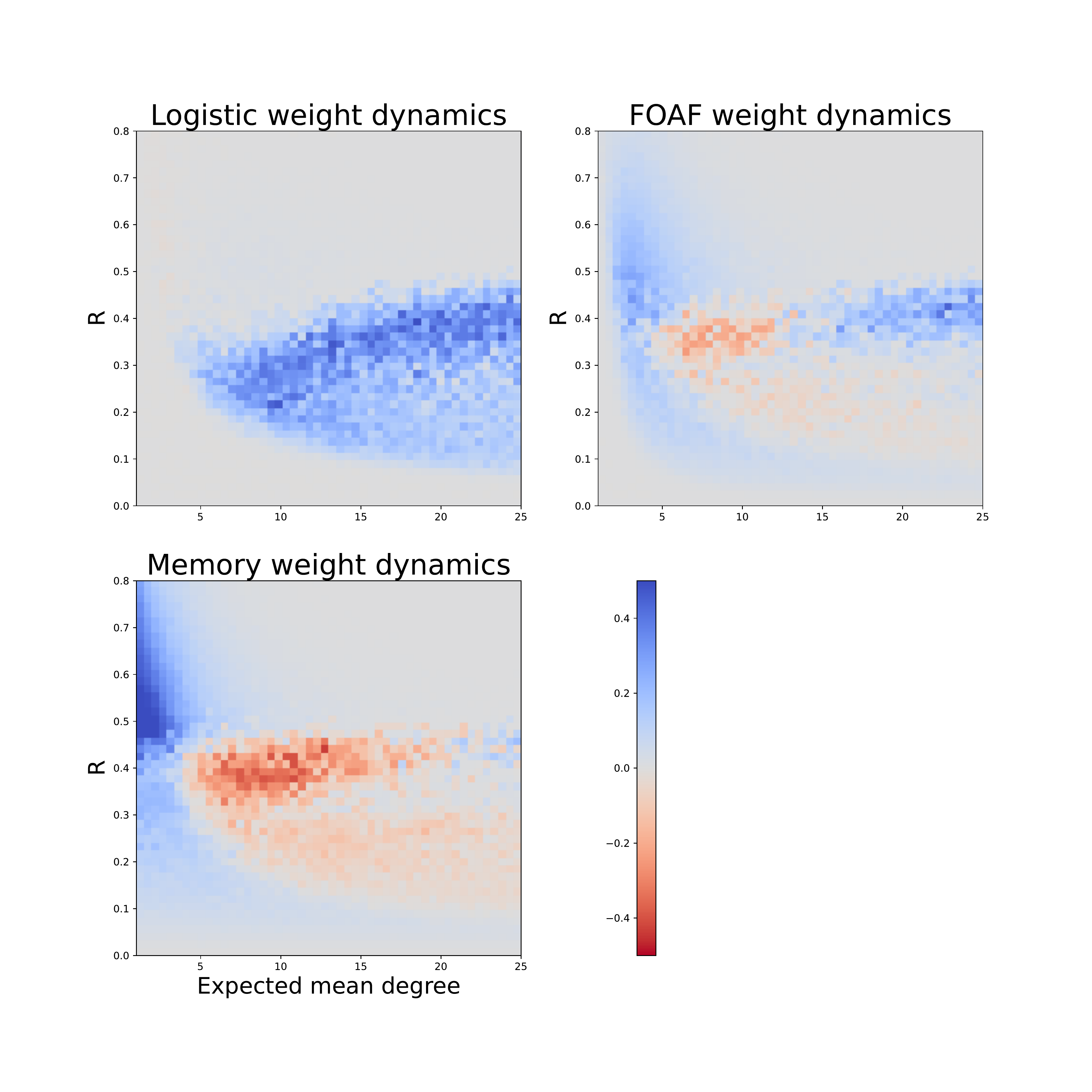}
    \caption{Heatmaps showing difference in the order parameter at steady state between the fixed network model \eqref{Eqn: fixed network ODEs} and dynamic network model \eqref{Eqn: general dynamic network system} with bounded confidence interaction function $\phi_R$ \eqref{Eqn: confidence indicator}. Each plot corresponds to a different type of weight dynamics. The expected mean degree of the initial network is varied along the x-axis of each heatmap and the bounded confidence radius $R$ is varied along the y-axis.}
    \label{fig:difference heatmap BC}
\end{figure}

\begin{figure}[ht!]
    \centering
    \includegraphics[width = \linewidth, trim = {3cm, 0cm, 3cm, 0cm}, clip]{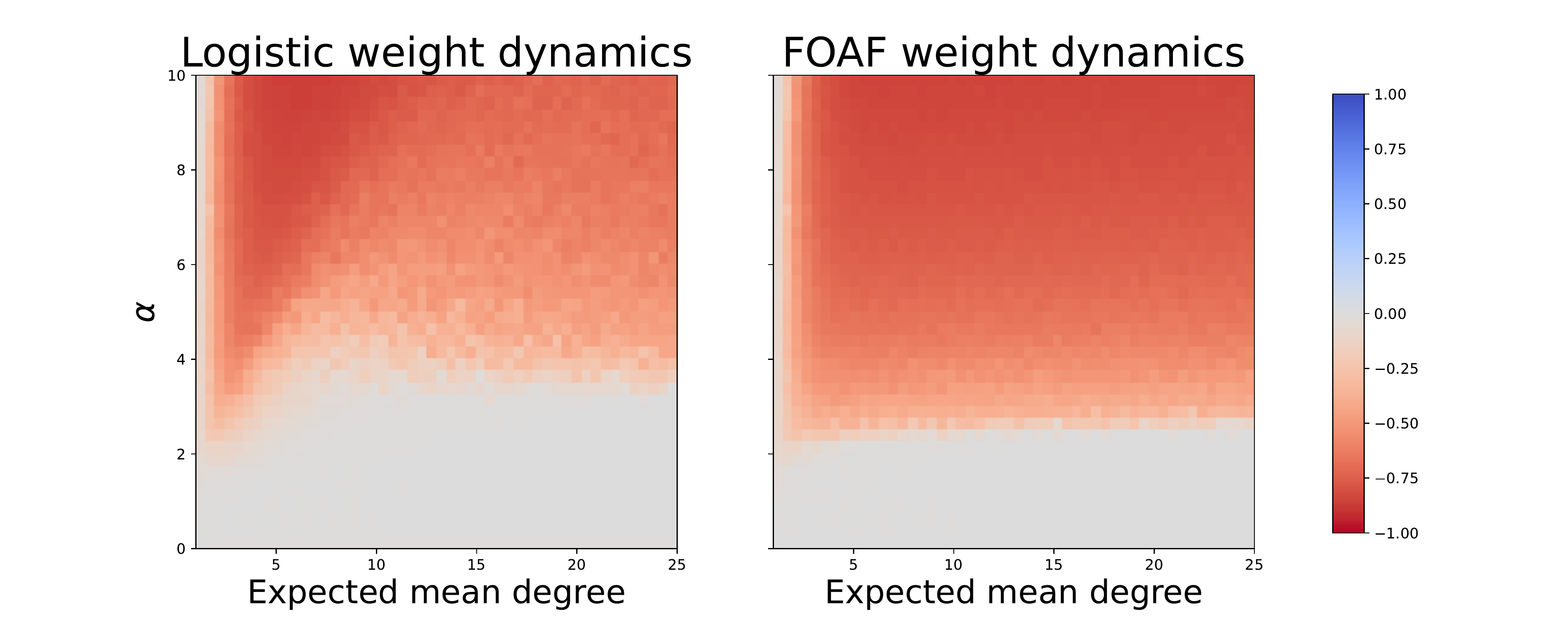}
    \caption{Heatmaps showing difference in the order parameter at steady state between the fixed network model \eqref{Eqn: fixed network ODEs} and dynamic network model \eqref{Eqn: general dynamic network system} with exponential interaction function $\phi_R$ \eqref{Eqn: confidence indicator}. Each plot corresponds to a different type of weight dynamics. The expected mean degree of the initial network is varied along the x-axis of each heatmap and the exponential decay rate $\alpha$ is varied along the y-axis.}
    \label{fig:difference heatmap exp}
\end{figure}

\clearpage
\section{Numerical simulation of extreme timescales} \label{Appendix: Simulation of extreme timescales}

To demonstrate the extreme timescales described in Theorem \ref{Theorem: slow weight dynamics} and Theorem \ref{theorem: fast memory weights} we perform a numerical simulation of the dynamic network model \eqref{Eqn: general dynamic network system} with memory weight dynamics \eqref{eqn: memory weight dynamics}, in which the weight dynamics \eqref{Eqn: general dynamic network system w_ij} is scaled by a factor of $\tau > 0$. We choose an interaction function $\phi$ which is a smoothed version of the bounded confidence interaction function $\phi_R$ with $R = 0.4$. This re-scaled system is
\begin{subequations} \label{eqn: dynamic network tau}
\begin{align}
    \frac{dx_i}{dt} &= \frac{1}{k_i} \sum_{j\neq i} w_{ij}\,\phi\big(|x_i - x_j|\big)\,(x_j - x_i) \,, \\ 
   % k_i &= \sum_{j=1}^N w_{ij} \,\\
    \frac{dw_{ij}}{dt} &= \tau\,\Big( \phi\big(|x_i - x_j|\big)\, 
 - w_{ij} \Big) \,.
 \end{align}
\end{subequations}
For $\tau \ll 1$ we approach the limit described by Theorem \ref{Theorem: slow weight dynamics} and compare against the fixed network model \eqref{Eqn: fixed network ODEs}. For $\tau \gg 1$ we compare against the `fast' limiting model \eqref{eqn: fast memory dynamics limiting model} from Theorem \ref{theorem: fast memory weights}. 

For a population size of $N = 50$ we fix initial opinions and an initial network. Fig.\ref{fig:comparing timescales timeseries} then shows the opinion formation process \eqref{eqn: dynamic network tau} for various values of $\tau$ at orders of magnitude from $\tau = 10^{-4}$ to $\tau = 10^2$. Note that $\tau = 1$ corresponds to the standard dynamic network model with memory weight dynamics. We also show the dynamics of the fixed network model \eqref{Eqn: fixed network ODEs} and the fast limiting model \eqref{eqn: fast memory dynamics limiting model} for comparison. We can see a clear change in the opinion formation process between the different values of $\tau$. Those small values at the top of the figure match closely to the behaviour of the fixed network model, while those larger values of $\tau$ at the bottom of the figure match the behaviour of the fast limiting model. Note that in Fig.\ref{fig:comparing timescales timeseries} time is shown on a log scale. 

\begin{figure}[ht!]
    \centering
    \includegraphics[width = 0.7\linewidth, trim = {2.5cm 8cm 3.5cm 8cm}, clip]{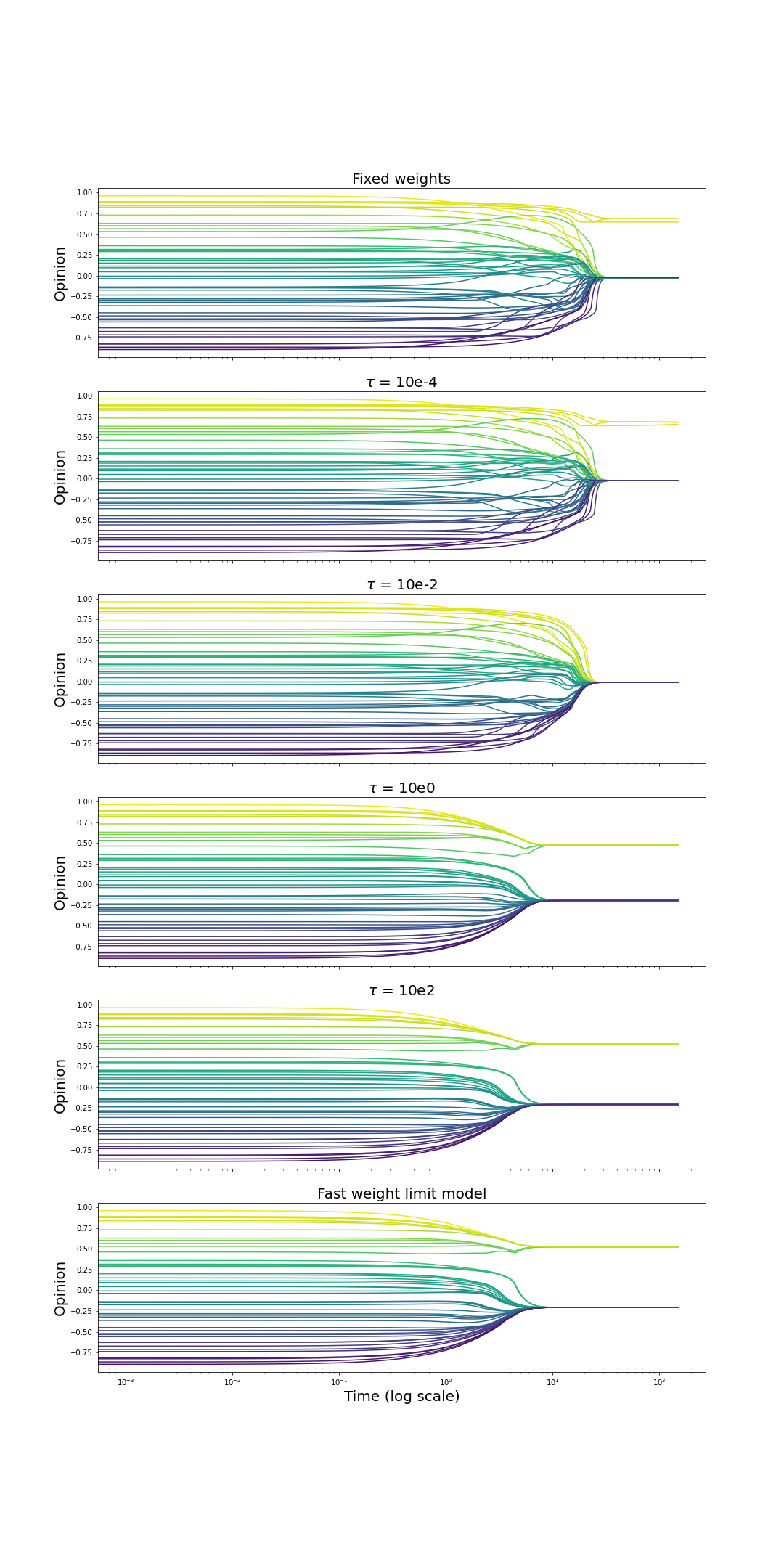}
    \caption{Demonstration of the results of Theorem \ref{Theorem: slow weight dynamics} and Theorem \ref{theorem: fast memory weights}. Dynamics of the fixed network model (top), dynamic network model (middle panels) and fast limit model (bottom) for identical initial opinions and network are shown. In the dynamic network model the speed of the memory weight dynamics has been scaled by $\tau$. For small values of $\tau$, towards the top of the figure, the dynamics resemble the fixed network model. For large values of $\tau$, towards the bottom of the figure, the dynamics resemble instead the fast limit model.}
    \label{fig:comparing timescales timeseries}
\end{figure}
The comparison between the different models is made more concrete in Fig.\ref{fig:comparing timescales errors}. For each value of $\tau$ we calculate the difference between that simulation and both the fixed network model and fast limit model. As all simulations use the same fixed timestep we calculate the average (over time) difference between the models' opinion states. It is clear that as $\tau$ becomes smaller the error against the fixed network model decreases and as $\tau$ becomes large the error against the fast limit model decreases. In Fig.\ref{fig:comparing timescales errors} both axes are on log scales. For all values of $\tau$ the model is simulated using a fixed timestep of $10^{-4}$, so the maximum value of $\tau$ considered is $10^3$. The minimum value considered is $10^{-5}$ to avoid encountering computational errors. 

\begin{figure}[ht!]
    \centering
    \includegraphics[width = \linewidth]{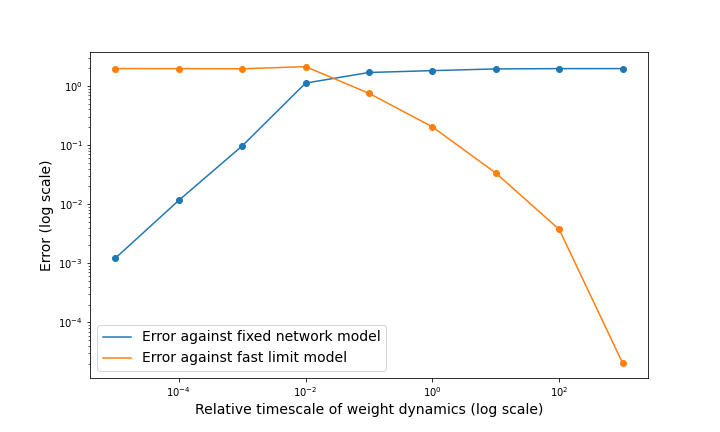}
    \caption{Comparing the re-scaled dynamic network model \eqref{eqn: dynamic network tau} against the fixed network model \eqref{Eqn: fixed network ODEs} and the fast limit model \eqref{eqn: fast memory dynamics result}. For each relative timescale the average (over time) difference between the models' opinion states is shown. When the weight dynamics are relatively slow (towards the left) the error against the fixed network model is low and when the weight dynamics are relatively fast (towards the right)  the error against the fast limit model is low.}
    \label{fig:comparing timescales errors}
\end{figure}

\clearpage
\section{Proof of Lemma \ref{Lemma: MotschTadmor}} \label{Appendix: Proofs}

\begin{proof} 

    The proof of this lemma is adapted from the argument presented in Theorem 4.3 of \cite{motsch2014heterophilious}, with the energy functional modified to include the fixed network $w$. 
    
    Define the energy functional 
    \begin{equation}
        \mathcal{E}(t) := \sum_{i,j} w_{ij} \,\varphi\big( |x_i(t) - x_j(t)| \big), \quad \varphi(r) := \int_0^r s \phi(s) \,ds
    \end{equation}
    This energy decreases over time since
    \begin{align*}
        \frac{d}{dt} \mathcal{E}(t) 
        % &= \sum_{i,j} w_{ij} \,\phi\big( |x_i - x_j| \big) \,\Big\langle\frac{dx_i}{dt} - \frac{dx_j}{dt}  \,, x_i - x_j \Big\rangle \,,\\
        &= -2 \sum_{i,j} w_{ij} \,\phi\big( |x_i - x_j| \big) \,\Big\langle \frac{dx_i}{dt} \,, x_i - x_j \Big\rangle \,,\quad\text{(by symmetry of $w$ and $\phi$)} \,\\
        &= -2 \sum_{i=1}^N \big\langle \frac{dx_i}{dt} \,, \sum_{j=1}^N w_{ij} \,\phi\big( |x_i - x_j| \big) \,(x_i - x_j) \big\rangle \,,\\
        &= -2 \sum_{i=1}^N k_i \Big|\frac{dx_i}{dt}\Big|^2 \leq 0 \,.
    \end{align*}
    The rate of this decrease can be quantified by considering
    \begin{align*}
        \Bigg( \frac{1}{2} \sum_{i,j} w_{ij} \,\phi\big( |x_i - x_j| \big) \,|x_i - x_j|^2 \Bigg)^2 
        &= \Bigg( \sum_{i,j} w_{ij} \,\phi\big( |x_i - x_j| \big) \,\langle x_i\,, x_j - x_i \rangle \Bigg)^2 \,,\\
        &= \Bigg( \sum_{i=1}^N \big\langle x_i \,, \sum_{j=1}^N w_{ij} \,\phi\big( |x_i - x_j| \big) \,(x_j - x_i) \big\rangle \Bigg)^2 \,,\\
        &= \Bigg( \sum_{i=1}^N k_i \,\Big\langle x_i \,, \frac{dx_i}{dt} \Big\rangle \Bigg)^2 \,,\\
        &\leq \Bigg( \sum_{i=1}^N k_i \,|x_i|^2 \Bigg)  \Bigg( \sum_{i=1}^N k_i \,\Big|\frac{dx_i}{dt}\Big|^2 \Bigg) \,,\\
        &\leq N^2 \Bigg( \sum_{i=1}^N k_i \,\Big|\frac{dx_i}{dt}\Big|^2 \Bigg) \,.
    \end{align*}
    Hence,
    \begin{equation}
        \frac{d}{dt} \mathcal{E}(t) \leq -\frac{1}{2 N^2} \Bigg( \sum_{i,j} w_{ij} \,\phi\big( |x_i - x_j| \big) \,|x_i - x_j|^2 \Bigg)^2 \,.
    \end{equation}
    As the energy $\mathcal{E}(t)$ is non-negative this implies that,
    \begin{equation}
        \int_0^\infty \Bigg( \sum_{i,j} w_{ij} \,\phi\big( |x_i(s) - x_j(s)| \big) \,\big| x_i(s) - x_j(s) \big|^2 \Bigg)^2 \,ds \, \leq 2 N^2 \mathcal{E}(0) < \infty \,.
    \end{equation}
    Hence the integrand must become arbitrarily small at some point in time, specifically for any $\epsilon > 0$ there exists a time $t^*$ such that
    \begin{equation}
        \sum_{i,j} w_{ij} \,\phi\big( |x_i(t^*) - x_j(t^*)| \big) \,\big| x_i(t^*) - x_j(t^*) \big|^2 < \epsilon \,. 
    \end{equation}
\end{proof}

\end{document}